\documentclass[11pt]{article}

\usepackage{fullpage}

\usepackage{algorithm}
\usepackage{algpseudocode}
\algtext*{EndWhile}{}
\algtext*{EndFor}{}
\algtext*{EndIf}{}

\usepackage{amssymb, amsmath, amsthm}
\usepackage{thm-restate}

\makeatletter
\@ifundefined{newcounteralias}{}{%
  \renewcommand{\thmt@autorefsetup}{%
    \@xa\def\csname\thmt@envname autorefname\@xa\endcsname
      \@xa{\thmt@thmname}%
  }%
}
\makeatother

\usepackage[colorlinks]{hyperref}
\usepackage[colorinlistoftodos, prependcaption, textsize=tiny]{todonotes}

\usepackage{xcolor}
\usepackage{graphicx}
\graphicspath{{figs/}}

\allowdisplaybreaks

\usepackage{enumitem}
\usepackage{multicol}
\usepackage{makecell}
\usepackage[table]{xcolor}

\usepackage{multirow}
\usepackage{hhline}

\declaretheorem[name=Theorem,numberwithin=section]{theorem}

\newtheorem{lemma}[theorem]{Lemma}

\newcommand{\av}{\mathrm{av}}

\newcommand{\R}{\mathbb{R}}

\newcommand{\calG}{{\mathcal{G}}}
\newcommand{\calI}{{\mathcal{I}}}
\newcommand{\calJ}{{\mathcal{J}}}

\newcommand{\ceil}[1]{\left\lceil #1\right\rceil}
\newcommand{\floor}[1]{\left\lfloor #1\right\rfloor}
\newcommand{\poly}{\mathrm{poly}}
\newcommand{\opt}{\mathrm{opt}}

\DeclareMathOperator{\union}{\bigcup}

\newcommand{\children}{\mathrm{children}}

\mathchardef\mhyphen="2D

\title{Simpler Algorithms for Capacitated $k$-Median with $(1 + \epsilon)$-Violations Using New Inequalities}

\date{}
\author{Han Dai \footnote{The work of HD and SL is supported by State Key Laboratory for Novel Software Technology, and New Cornerstone Science Foundation.} \\School of Computer Science,\\ Nanjing University, \\ Nanjing, Jiangsu Province, China. \\ \href{mailto:handai@smail.nju.edu.cn}{handai@smail.nju.edu.cn}
 \and
Shi Li \footnotemark[\value{footnote}] \\
School of Computer Science,\\ Nanjing University, \\ Nanjing, Jiangsu Province, China. \\ \href{mailto:shili@nju.edu.cn}{shili@nju.edu.cn}
}

\begin{document}
     \maketitle

\begin{abstract} We revisit the capacitated $k$-median problem, focusing on constant-factor approximation algorithms with either $(1+\epsilon)$-cardinality or capacity violations. Previous results \cite{DemirciL16, Li16} achieve an $O_\epsilon(1)$-approximation with $(1+\epsilon)$-violations via complex extended configuration linear programming formulations with the round-or-cut framework.

   We propose a simpler, unified round-or-cut framework that handles both cardinality and capacity violation settings.  Our main technical contribution is a mechanism that directly utilizes the costs of local instances to generate valid inequalities. Specifically, we partition the metric space into local pieces of size $O(1/\epsilon)$ as in previous work, and construct instances of capacitated clustering with penalties, which allow client demands to be unserved at a penalty cost. Solving these local instances yields lower bounds on the original LP cost, yielding valid LP constraints.  In contrast, previous approaches rely on complex extended configuration LPs over connection variables $x_{ij}$'s and opening variables $y_i$'s. 
   
   As a secondary contribution, our approach achieves an improved running time of $\poly(n, 1/\epsilon)$, bypassing the exponential dependence on $1/\epsilon$ required to solve the previous extended configuration LPs. Instead of solving the local instances optimally in time exponential in $1/\epsilon$, we obtain $O(1/\epsilon)$-approximate solutions in $\poly(n, 1/\epsilon)$ time.  Finally, for the cardinality-violation setting, our algorithm opens a set of at most $(1+\epsilon)k$ \emph{distinct} facilities, and thus gives a bicriteria approximation in the \emph{hard-capacitated} setting. In contrast, the algorithm of \cite{Li16} may open the same facility multiple times.
\end{abstract}

\section{Introduction.}
In the capacitated $k$-median (CKM) problem,  we are given a facility set $F$, where every facility $i $ is associated with a capacity $u_i\in \mathbb{Z}_{\geq 0}$, a client set $C$, an integer $k\in 
\mathbb{Z}_{\geq 0}$, and a metric $d$ over $F\cup C$.  The goal is to select a subset $S\subseteq F$ of at most $k$ open facilities and a connection assignment $\sigma: C\to S$  such that $|\sigma^{-1}(i)|\leq u_i$ for every facility $i \in S$, and so as to minimize the connection cost $\sum_{j\in C}d(j, \sigma(j))$.

When $u_i = \infty$ for all $i \in F$, the CKM problem reduces to the classic (uncapacitated) $k$-median problem. Following the first constant-factor approximation algorithm by Charikar, Guha, Tardos and Shmoys~\cite{Charikar2002}, a long line of work has progressively improved the approximation ratio~\cite{guha1999greedy, hochbaum1986unified, siamcomp/LiS16, ByrkaPRST15, ByrkaPRST17, GowdaPST23, CharikarCGG98, AryaGKMMP04, Cohen-AddadGHOS22, KanungoMNPSW04, KorupoluPR00}, including recent $2+\epsilon$~\cite{cohen20252+, ByrkaLi2026}. On the negative side, it is known to be NP-hard to approximate the problem within a factor better than $1+2/e$~\cite{JainMMSV03,JainMS02}.

In contrast, obtaining a true constant-factor approximation for the capacitated $k$-median (CKM) problem remains a notoriously hard problem in approximation algorithms. The current best approximation ratio is $O(\log k)$ \cite{Adamczyk2019}, achieved via the metric tree-embedding framework of Fakcharoenphol, Rao, and Talwar~\cite{DBLP:conf/stoc/FakcharoenpholRT03, DBLP:journals/jcss/FakcharoenpholRT04}. Consequently, most prior work focuses on bicriteria approximation algorithms, which allow constant-factor violations on either the number of open facilities (the cardinality constraint) or the facility capacity constraints.

For the uniform capacitated $k$-median problem—where all facilities have the same capacity—Charikar, Guha, Tardos and Shmoys~\cite{CharikarGTS99} gave a $16$-approximation with a capacity violation factor of $3$. Later, Li~\cite{Li14} presented an $O(1/\epsilon)$-approximation with a $(2 + \epsilon)$-capacity violation. For the general capacity setting, Chuzhoy and Rabani~\cite{ChuzhoyR05} obtained a $40$-approximation with a capacity violation factor of $50$. Byrka, Fleszar,  Rybicki, and Spoerhase~\cite{ByrkaFRS15} subsequently achieved a constant-factor approximation with a $(3+\epsilon)$-capacity violation. On the cardinality-violation front,  Aardal,  Berg,  Gijswijt  and Li~\cite{AardalBGL15} designed a $(7+\epsilon)$-approximation algorithm that opens at most $2k$ facilities (a factor-$2$ cardinality violation). 

Most bicriteria approximations for CKM rely on linear programming (LP) relaxations. However, the natural LP relaxation has an unbounded integrality gap, even if one permits up to a $(2-\epsilon)$-violation on either the cardinality or capacity constraint. Thus, to achieve an $O(1)$-approximation via standard LP formulations, a violation factor of at least $2$ is necessary—a barrier that holds even for uniform capacities.

Li~\cite{li2017uniform} was the first to break this factor-$2$ barrier, proposing an $\exp(O(1/\epsilon^2))$-approximation algorithm for uniform CKM with a $(1+\epsilon)$-cardinality violation. To overcome the integrality gap, Li introduced an exponential family of \emph{rectangle inequalities}. Because separating these inequalities efficiently is non-trivial, Li employed a round-or-cut framework: given a fractional solution that only satisfies the basic constraints, the algorithm either rounds it to an integer solution with cost at most $\alpha$ times the value of the LP solution (round), or identifies a violated rectangle inequality (cut). Coupled with the Ellipsoid method, this yields a polynomial-time $\alpha$-approximation. Li~\cite{Li16} subsequently extended this framework to general capacities, achieving an $O(\frac{1}{\epsilon^2} \log \frac{1}{\epsilon})$-approximation under a $(1+\epsilon)$-cardinality violation—providing an exponential improvement in the approximation ratio's dependence on $\epsilon$. Since rectangle constraints fail in the general capacity setting, the approach in~\cite{Li16} instead constructed exponentially many extended configuration LP relaxations with auxiliary variables.

Similarly, in the capacity-violation setting, Byrka, Rybicki, and Uniyal~\cite{ByrkaRU16} developed an $O(1/\epsilon^2)$-approximation algorithm for uniform CKM with a $(1+\epsilon)$-capacity violation. Demirci and Li~\cite{DemirciL16} generalized this result to arbitrary capacities, achieving an $O(1/\epsilon^5)$-approximation with a $(1+\epsilon)$-capacity violation using extended configuration LPs in~\cite{Li16}. Together, these results establish optimal $(1+\epsilon)$-violation factors in both settings, leaving the existence of a true $O(1)$-approximation as the primary remaining open problem.

Regarding the running time, the uniform-capacity algorithm of \cite{li2017uniform} runs in $\poly(n, 1/\epsilon)$ time, whereas the uniform-capacity algorithm of \cite{ByrkaRU16} and the general-capacity algorithms of \cite{Li16} and~\cite{DemirciL16} run in $n^{\poly(1/\epsilon)}$ time.
Furthermore, there is a distinction between the \emph{soft-capacitated} setting (where multiple copies of a facility may be opened, each counting toward $k$) and the \emph{hard-capacitated} setting (where each facility can be opened at most once). Notably, the state-of-the-art results for the cardinality violation setting~\cite{Li16} apply only to the soft-capacitated variant.

\subsection{Our Results}

We present a simplified and unified bicriteria approximation framework for the capacitated $k$-median problem with hard capacities, addressing the $(1+\epsilon)$-cardinality and $(1+\epsilon)$-capacity violation settings. Throughout, $n$ denotes the total number of facilities and clients, and $\opt$ denotes the optimum cost of the CKM instance, without any violation. 

\begin{theorem}
    \label{thm:main-card}
    For any constant $\epsilon > 0$, there exists a $\poly(n, 1/\epsilon)$-time algorithm that, given a  hard capacitated $k$-median instance, outputs a solution opening at most $(1+\epsilon)k$ facilities while satisfying all capacity constraints, with cost at most $O(1/\epsilon^5) \cdot \opt$.
\end{theorem}

Our second main result addresses the capacity-violation setting.

\begin{theorem}
    \label{thm:main-cap}
    For any constant $\epsilon > 0$, there exists a $\poly(n, 1/\epsilon)$-time algorithm that, given a  hard capacitated $k$-median instance, outputs a solution opening at most $k$ facilities while connecting at most $\lceil(1+\epsilon)u_i\rceil$ clients to each open facility $i$, with cost at most $O(1/\epsilon^6) \cdot \opt$.
\end{theorem}

\begin{table}[htp]
    \centering
    \renewcommand{\arraystretch}{1.4}
    \setlength{\tabcolsep}{7pt}
    \setlength{\arrayrulewidth}{0.4pt}
    \arrayrulecolor{black}

    \resizebox{\linewidth}{!}{%
        \begin{tabular}{c|c|c|c|c|c}
            \hline
            
\multicolumn{6}{c}{%
                \textbf{Capacitated $k$-Median}
            } \\
            \hline
            \makecell{\textbf{Capacity}}
            & \makecell{\textbf{Approx. factor}}
            & \makecell{\textbf{Facility violation}} &\makecell{\textbf{Capacity violation}}
            & \textbf{Run time}
            & \textbf{Reference} \\
            \hline

        Soft
             & $O\!\left(\epsilon^{-2}\log(1/\epsilon)\right)$
            & $(1+\epsilon)k$
            & $-$
            & $n^{\poly(1/\epsilon)}$
            & Li~\cite{Li16} \\
            \hline

            & $O(\epsilon^{-5})$
            & $k$
            & $1+\epsilon$
            & $n^{\poly(1/\epsilon)}$
            & Demirci and Li~\cite{DemirciL16} \\
            \hhline{~|-|-|-|-|-}

            Hard
            & $O(\epsilon^{-5})$
            & $(1+\epsilon)k$
            & $-$
            & $\operatorname{poly}(n,1/\epsilon)$
            & \textcolor{blue}{Theorem~\ref{thm:main-card}} \\
            \hhline{~|-|-|-|-|-}

            & $O(\epsilon^{-6})$
            & $k$
            & $1+\epsilon$
            & $\operatorname{poly}(n,1/\epsilon)$
            & \textcolor{blue}{Theorem~\ref{thm:main-cap}} \\
            \hline
        \end{tabular}  
                }
    \caption{%
        For Theorem~1.2, the capacity bound in the final integral
        assignment is $\lceil(1+\epsilon)u_i\rceil$
        for each open facility $i\in F$.
    }
    \label{tab:results-comparison}
\end{table}

Compared with \cite{Li16, DemirciL16}, our main contribution is to generate valid inequalities from local penalty instances, thereby avoiding the explicit use of extended configuration LPs and improving the running time to $\poly(n, 1/\epsilon)$.

\subsection{Our Techniques}
\label{sec:techniques}
We now present the key ideas behind our simplified framework. To illustrate our core intuition, consider the integrality gap of the basic LP relaxation (LP\eqref{LP:basic}) in the cardinality-violation setting. Suppose we have two facilities, $i$ and $i'$, and $u+1$ clients co-located at the same point, where $u_i = u_{i'} = u$ for some large integer $u$. The LP can fractionally open $1 + 1/u$ facilities from $\{i, i'\}$ to serve all $u + 1 = (1 + 1/u)u$ clients at cost $0$. By taking disjoint copies of this gadget, the LP can open $k$ fractional facilities in total at cost $0$. However, any integer solution must open $2k / (1 + 1/u)$ facilities to achieve cost $0$, requiring a cardinality violation factor of $2/(1+1/u)$, which approaches $2$ for large $u$.

More generally,  a CKM instance can consist of far-apart ``islands'', where the LP solution opens $z$ fractional facilities per island for some small integer $z \ll 1/\epsilon$. Under a $(1+\epsilon)$-cardinality budget, we cannot afford to open $z + 1$ facilities on each island. Thus, one must guarantee that each island's subproblem admits a small-cost integral solution with at most $z$ open facilities. While \emph{rectangle inequalities} suffice for uniform capacities~\cite{li2017uniform}, the general capacity setting becomes significantly more intricate: Li~\cite{Li16} introduced auxiliary variables to encode every valid configuration of open facilities on each island (up to size $\Theta(1/\epsilon)$), leading to a heavy configuration LP. The capacity violation setting exhibits a similar complexity behavior. 

We take a simpler, direct approach. Rather than constructing complex extended LPs, we solve the local instance on each island directly to establish a valid lower bound on its client connection cost, and then enforce this bound via a valid linear constraint. To handle the more general case where a small fraction of demands from clients may be satisfied outside their local island, we formalize each local subproblem as a \emph{capacitated clustering with penalty} problem, where unserved client demand incurs a penalty proportional to the distance from the island to outside. Furthermore, by approximating these local penalty instances within an $O(1/\epsilon)$ factor instead of solving them optimally, we reduce the running time from $n^{O(1/\epsilon)}$ to $\text{poly}(n, 1/\epsilon)$.

Our island construction follows the spanning tree framework of prior work~\cite{li2017uniform, Li16, ByrkaRU16, DemirciL16}. Given a fractional solution to the basic LP, a filtering step selects a set of representative clients, and all demands are moved to representatives at a small cost. We then run Kruskal's algorithm on these representatives to build a Minimum Spanning Tree (MST) and color its edges in black, gray, and white. The resulting \emph{black components} (maximal connected subtrees of black edges) form our islands. For each black component, we define a family of capacitated clustering instances with penalty, parameterized by the number $t$ of open facilities. We solve these instances approximately and use the resulting lower bounds to generate valid LP constraints. As in previous work, the construction of black components depends on the LP solution. Thus these valid constraints cannot be fixed a priori, necessitating the round-or-cut framework.

Finally, our rounding procedures for both the cardinality- and capacity-violation settings follow a similar framework. We first construct a local fractional solution for each black component by combining the LP solution with our local penalty solutions. Next, we partition the black components into \emph{groups}, each containing $\Theta(1/\epsilon)$ fractional open facilities, ensuring that demand within each group can be satisfied internally with a $(1+\epsilon)$-violation. We then formulate an auxiliary LP over these local solutions. The structure of this auxiliary LP guarantees that in an extreme solution,  each group contains $O(1)$ strictly fractional facility variables. A $(1+\epsilon)$-violation provides just enough slack to round these residual fractional facilities to integers at low cost. The specific constructions of local solutions and group partition schemes vary between the two violation settings  and are motivated by previous work.

 \paragraph{Organization} 
 In Section~\ref{sec:Common-algorithm}, we give a common algorithm for cardinality-violation and capacity-violation settings.  In Section~\ref{sec:Cardinality-Violation}, we give the $(1/\epsilon^5)$-approximation for the capacitated $k$-median problem with $(1+\epsilon)k$ open facilities,  proving Theorem~\ref{thm:main-card}.  
In Section~\ref{sec:Capacity-Violation}, we give the $(1/\epsilon^6)$-approximation for the capacitated $k$-median problem with  at most $k$ facilities while
connecting at most $\lceil(1 + \epsilon)u_i\rceil $ clients to each open facility, proving Theorem~\ref{thm:main-cap}. 
All missing proofs can be found in  Appendix \ref{sec:appendix}.
      \section{Common algorithm for cardinality-violation and capacity-violation settings}\label{sec:Common-algorithm}

In this section, we describe the common procedure in our rounding algorithms for both the cardinality and capacity-violation settings.   The following is the basic LP relaxation for both settings:
\begin{align}
      \min \qquad \sum_{j \in C} \phi_j\label{LP:basic}
\end{align}\vspace*{-30pt}

\noindent\begin{minipage}[t]{0.40\textwidth}
 \begin{align}
      \sum_{i\in F} x_{i,j} &=1 ,  &\quad &\forall j\in C \label{LPC:connection-1-facility}\\
      \sum_{i\in F} y_i &\leq k, \label{LPC:open-k-facility}\\
      x_{i, j}&\leq y_i, &\quad &\forall i\in F,  j\in C \label{LPC:x<y}
 \end{align}
\end{minipage}\hfill
\begin{minipage}[t]{0.60\textwidth}
 \begin{align} 
 \sum_{j\in C} x_{i, j} &\leq u_i y_i , &\quad &\forall i\in F \label{LPC:capacity}\\
 \phi_j &\geq \sum_{i \in F} d(i, j)  x_{ij},  &\quad &\forall j \in C \label{LPC:phi-j}\\
x_{i,j}, y_i&\in [0, 1], &\quad &\forall i\in F, j\in C \label{LPC:x-y-non-negative}
  \end{align}
\end{minipage}\bigskip 

In the corresponding integer program, $x_{ij}\in \{0, 1\}$ for every facility $i \in F$ and client $j \in C$ indicates whether $j$ is connected to $i$, $y_i\in \{0, 1\}$ for $i\in F$ indicates whether $i$ is open, and $\phi_j$ for $j \in C$ denotes the connection cost of the client $j \in C$.  The objective \eqref{LP:basic} is the sum of connection costs over all clients. \eqref{LPC:connection-1-facility} requires that every $j$ is connected to exactly 1 facility, \eqref{LPC:open-k-facility} requires that the number of open facilities is at most $k$, and \eqref{LPC:x<y} says that $j$ can be connected to $i$ only if $i$ is open.  \eqref{LPC:capacity} is the capacity requirement, and \eqref{LPC:phi-j} is by the definition of $\phi_j$. In the LP relaxation program, we only require $x_{ij}$'s and $y_i$'s to be in $[0, 1]$, as stated in \eqref{LPC:x-y-non-negative}.

We shall use the ''round-or-cut'' paradigm \cite{DBLP:conf/soda/CarrFLP00, DBLP:journals/siamcomp/AnSS17}. Given a fractional solution $(x^*, y^*, \phi^*)$ to LP\eqref{LP:basic}, our rounding algorithm either succeeds by returning an integral solution whose cost is at most $O_\epsilon(1)\cdot \sum_{j \in C}\phi^*_j$ (which may violate the cardinality constraint or the capacity constraints by a factor of $(1+\epsilon)$),  or returns a valid constraint of the following form, that is violated by the solution $(x^*, y^*, \phi^*)$: 
	\begin{align}
		s \sum_{i \in F'}y_i + \sum_{j \in C} w_j \phi_j  \geq  \Phi, \label{LPC:penalty}
	\end{align}
	where $s \in \R_{\geq 0}$, $F' \subseteq F$, $w_j \in [0, 1]$ for every $j \in C$, and $\Phi \in \R_{\geq 0}$.  We can think of that each facility $i \in F'$ has a cost $s$, and a client $j \in C$ has weight $w_j \in [0, 1]$. The inequality says that the sum of facility cost and the weighted connection cost of clients is at least $\Phi$. When our rounding algorithm returns a constraint of this form, we guarantee its validity.  That is, it is satisfied by all integral solutions to the CKM instance.  It is well-known that such an algorithm is sufficient to obtain an $O_\epsilon(1)$-approximation, with $(1+\epsilon)$-cardinality or capacity violation. 

    Furthermore, when the facility opening variables $y_i$'s are integral ($y_i \in \{0, 1\}$ for every $i \in F$), the optimum client-to-facility assignment reduces to a minimum-cost maximum flow (or matching) problem. As allowing $x_{ij}$'s to take fractional values does not decrease the cost, our rounding algorithm allows the connections between clients and facilities to be fractional. We assume every $j \in C$ has 1 unit of demand. Our goal is to choose a set of at most $k$ open facilities, and send the demands to the open facilities, so that every open facility $i$ receives at most $u_i$ units of demand.  When a $(1+\epsilon)$-capacity violation is allowed, an open facility $i$ can take up to $(1+\epsilon)u_i$ units of demand fractionally. But in an integral solution, we need to increase its capacity to $\ceil{(1+\epsilon)u_i}$.

	\medskip

   We can assume $k\le |F|$; otherwise, opening all facilities and
 computing a minimum-cost assignment is optimal.
 Fix a solution $(x^*, y^*, \phi^*)$ to LP\eqref{LP:basic}.
 Increase $y^*$ coordinatewise within $[0, 1]^F$ until
 $\sum_{i\in F}y_i^*=k$, keeping $x^*$ and $\phi^*$ unchanged.
 This preserves feasibility and cost. Any constraint of form \eqref{LPC:penalty}
 violated afterward was also violated before  this modification,
 since its coefficients of $y$ are nonnegative.
 We continue to denote the modified vector by $y^*$.

           For any set $F' \subseteq F$ of facilities and any set $C' \subseteq C$ of clients, we shall let $x^*(F', C') : = \sum_{i\in F', j\in C'} x^*_{i, j}$;  we simply write $x^*(i, C')$ for $x^*(\{i\}, C')$ and $x^*(F', j)$ for $x^*(F', \{j\})$.  We define $y^*(F') =\sum_{i\in F'} y^*_i$ for any $F' \subseteq F$. This notational convention applies to other variants of $y$ as well.

    Throughout the paper, we let $\ell = \Theta(1/\epsilon)$ be an integer with a sufficiently large hidden constant in $\Theta(\cdot)$.

    \subsection{Representatives $R$ and Initial Moving of Demands} \label{Initial-Moving-of-Demands}
        We use the same framework as in \cite{li2017uniform, Li16, DemirciL16, ByrkaRU16} to choose a set $R$ of representatives and move all demands to $R$. 
       We define the average connection cost of each client $j\in C$ as $d_{\av}(j) := \sum_{i\in F} x^*_{i,j} d(i,j)$. The value of the solution $(x^*, y^*, \phi^*)$ is $\phi^*\geq \sum_{i\in F, j\in C}x^*_{i, j}d(i, j)=\sum_{j\in C}d_{\av}(j)$.
   We select representatives $R$ using Algorithm \ref{alg:Filtering-of-Representatives}.
\begin{algorithm}
	\caption{filtering of representatives}
	\label{alg:Filtering-of-Representatives}
	\begin{algorithmic}[1]
    \State \textbf{Input:} a facility set $F$, a client set $C$, a metric $d$,
                 and a fractional solution $(x^*, y^*, \phi^*)$.
                
           \State      \textbf{Output:} a set $R \subseteq C$ of representatives.
                 
                 \State  compute
               $d_{\av}(j)
               \gets \sum_{i\in F} x^*_{i, j}d(i, j)$
               for each $j\in C$.
        \State  $R \gets \emptyset$;
               $C' \gets C$.
		\While{$C' \neq \emptyset$}
            \State  choose
                   $v \in \arg\min_{j\in C'}
                   d_{\av}(j)$,
                   breaking ties arbitrarily.
            \State  $R \gets R \cup \{v\}$; \, $C'\gets C'
                   \setminus
                   \{j\in C':
                     d(j, v)\le 4d_{\av}(j)\}$. \Comment{$v$ itself is also removed}
       \EndWhile
        \State \Return $R$. 
	\end{algorithmic}
\end{algorithm}

 We perform a Voronoi partition on facility set $F$ centered at representative set $R$: each facility $i\in F$ is assigned to its nearest representative $v\in R$, yielding a partition of $F$ into clusters $\{U_v: v\in R\}$, where $U_v=\{i\in F: v=\arg\min_{v'\in R} d(i, v')\}$.  For any subset $V\subseteq R$ of representatives, we use $U(V) := \bigcup_{v\in V} U_v$ to denote the union of Voronoi regions with centers $V$.

\begin{restatable}{lemma}{lemmarepresentatives}
        \label{lemmarepresentatives}
            The following statements hold:
             \begin{enumerate}[label=(\thelemma\alph*),leftmargin=*]
			\item for all  $v, v'\in R $ and $v\neq v'$,  we have $d(v, v')>4\max\{d_{\av}(v), d_{\av}(v')\}$; \label{represen-distance}
            \item for all $j \in  C$, there exists $v \in R$, such that $d_{\av}(v) \leq d_{\av}(j)$ and $d(v, j) \leq 4d_\av(j)$; \label{j-removed}
            \item  $y^*(U_v) \geq 1/2$ for every $v \in R$; \label{opening-of-U_v}
             \item  for any $v\in R, i\in U_v$ and $j\in C$, we have $d (i, v)\leq d(i, j)+ 4d_\av(j)$.\label{le:move-distance}
		\end{enumerate}	
        \end{restatable}

         We move all demands to client representatives. First, for each client $j \in C$ and $i \in F$, we move $x_{i, j}$ units of demand from $j$ to $i$. Obviously, the moving cost is exactly $\sum_{i\in F, j\in C}x^*_{i, j} d(i, j)$.  After this step, each $i \in F$ has $\sum_{j\in C} x^*_{i, j} = x^*({i, C}) \leq u_i y_i$ units of demand. Second, for each $v \in R$ and $i \in  U_v$, we move the $x^*({i, C})$ units of
      demand at $i$ to $v$. Then, all demands are moved to $R$ and a representative $v \in R$  has $\sum_{i\in U_v} x^*(i, C) = x^*({U_v, C})$ units of demand. 
      
      Let $g_{j, v}$ be the amount of demand moved from a $j \in C$ to $v \in R$. That is, we define $g_{j, v} = x^*({U_v, j})$. Clearly, for every $j \in C$, we have $\sum_{v \in R}g_{j, v} = x^*({F, j}) = 1$.

     Let $D_i: =\sum_{j\in C}x^*_{i, j}(d(i, j)+d_\av(j))$ for each $i\in F$, and $D(F'): =\sum_{i\in F'}D_i$ for every $F'\subseteq F$. Notice that $D(F) = \sum_{i \in F} D_i = \sum_{i \in F, j \in C} x^*_{i,j}d(i,j)+\sum_{j \in C}d_\av(j) \leq 2\sum_{j \in C}\phi^*_j$. We can think of $D_i$ as the budget given to $i$.

      \begin{lemma}
        \label{lemma:x*iCdiv}
          For every $v \in R$, we have $\sum_{i\in U_v} x^*({i, C}) d(i, v) \leq O(1)\cdot D(U_v)$.
      \end{lemma}
      \begin{proof}
          By Property \ref{le:move-distance}, we have $d (i, v)\leq d(i, j)+ 4d_\av(j)$ for every $i\in U_v$ and $j\in C$. Thus, 
          $\sum_{i\in U_v} x^*({i, C}) d(i, v) \leq \sum_{i\in U_v, j\in C}x^*_{i, j}(d(i, j)+4 d_\av(j))\leq 4\sum_{i\in U_v}D_i\leq O(1)D(U_v).$ 
      \end{proof}

This immediately implies the following lemma, which shows that the initial moving of demands from clients $C$ to representatives $R$ is small.
    \begin{lemma}\label{lem:representative}
   For every $v \in R$, we have $\sum_{j\in C} g_{j, v} d(j, v)\leq O(1) \cdot D(U_v)$.
    \end{lemma}
    \begin{proof}
        For each $i\in U_v$, by the triangle inequality, $d(j, v)\leq d(j, i)+d(i, v)$.  So
        \begin{align*}
            \sum_{j\in C}g_{j, v} d(v, j)&=\sum_{j\in C, i\in U_v}x^*_{i, j}  d(v, j)\leq \sum_{j\in C, i\in U_v}x_{i, j}^* (d(j, i)+d(i, v))\\
            & \leq O(1)\cdot D(U_v).
        \end{align*}
        The last inequality follows from Lemma~\ref{lemma:x*iCdiv}.
    \end{proof}

 \subsection{Construction of Black Components and Gray Forest Using Kruskal's Algorithm} 
 We now describe the framework of \cite{Li16, DemirciL16} in which we run Kruskal's Algorithm to obtain a minimum spanning tree (MST) for $R$, and color the edges of the MST using black, gray and white colors. We can assume $k \geq \ell$. Otherwise, the $O(\log k)$-approximation algorithm for capacitated $k$-median problem by \cite{Adamczyk2019} gives a true approximation ratio of $O(\log k) = O(\log 1/\epsilon)$.

 For any $v\in R$,  define its weight as $ y^*({U_v})= \sum_{i\in U_v}y^*_i$ satisfying $y^*({U_v})\geq \frac{1}{2}$.  For any subset  $J \subseteq R$, define its weight as $y^*(U(J))=\sum_{v \in J} y^*({U_v})$. If $y^*(U(J)) \geq \ell$, then $J$ is said to be heavy; otherwise, it is said to be light.

    Recall that in Kruskal's algorithm for MST, we maintain the edge set $E_{\mathrm{MST}}$  and a partition $\mathcal{P}$ of $R$, defined in Algorithm \ref{alg:kruskal-coloring}.

\begin{algorithm}[h]
    \caption{Kruskal's algorithm with edge coloring}
    \label{alg:kruskal-coloring}
    \begin{algorithmic}[1]
        \State \textbf{Input:} a representative set $R$, a metric $d$,
                 weights $\{y^*(U_v)\}_{v\in R}$,
                 and a threshold $\ell$.
        \State \textbf{Output:} a minimum spanning tree  on $R$ with colored edges.

        \State $E_{\text{MST}}\gets\emptyset$;
               $\mathcal{P}\gets\{\{v\}: v\in R\}$.
        \State sort the edges in $\binom{R}{2}$ in
               non-decreasing order of length,
               breaking ties arbitrarily.

        \ForAll{$e=\{v, v'\}$ in the sorted order}
            \State let $J_v, J_{v'}\in\mathcal{P}$ be the
                   components containing $v$ and $v'$, respectively.

            \If{$J_v\neq J_{v'}$}
                \State $E_{\text{MST}}
                       \gets E_{\text{MST}}\cup\{e\}$.

                \If{both $J_v$ and $J_{v'}$ are light}
                    \State color $e$ black and leave it undirected.
                \ElsIf{exactly one of $J_v$ and $J_{v'}$ is light}
                    \State color $e$ gray and direct it from the light set to the heavy set.
                \Else
                    \State color $e$ white and leave it undirected.
                \EndIf

                \State $\mathcal{P}
                       \gets
                       \bigl(\mathcal{P}\setminus\{J_v,J_{v'}\}\bigr)
                       \cup\{J_v\cup J_{v'}\}$.
            \EndIf
        \EndFor
        \State \Return $(R, E_{\mathrm{MST}})$, with all edges colored and gray edges directed.
    \end{algorithmic}
\end{algorithm}
  Let $(R, E_{\text{MST}})$ be the minimum spanning tree returned by Algorithm  \ref{alg:kruskal-coloring}.
    Consider the forest formed by the black edges of the MST. The vertex set of any tree in the forest is said to be a \textit{black component}, or simply \emph{component}. Let $\mathcal{J}$ be the set of all black components; so, $\mathcal{J}$ forms a partition of $R$. For every $J \in \mathcal{J}$, we define $L_J := d(J, R \setminus J)$.
    When $J= R$, define $L_J: =2\cdot\text{diam}(F\cup C)$.
    
    \begin{restatable}{claim}{claimcomponentJ}
    \label{claim:componentJ}
        For each $J\in \mathcal J $ and $v, v'\in J$,  each black edge $e=(v, v')$ in $J$ satisfies  $d(v, v')\leq L_J$.
    \end{restatable}\medskip

    We take the MST, remove all white edges, and contract all black edges. The resulting graph is a forest, which we call the \emph{gray forest} and denote by $\Upsilon$. We call each tree in the forest a gray tree. We keep the direction of the gray edges, as defined earlier. So, the gray forest is over the set of black components, and each gray tree is a rooted tree with gray edges directed towards the root.

   \begin{restatable}{lemma}{lemmagrayforest}\label{lemmagrayforest}
          For each tree $\tau$ in the gray forest $\Upsilon$, the following properties hold:
         \begin{enumerate}[label=(\thelemma\alph*),leftmargin=*]
            \item the root of $\tau$ is a heavy black component, and all other vertices of $\tau$ are light black components. \label{le:weight-J}
            \item let $J, J'$ be two black components in $\tau$ such that $J'$ is the parent of $J$. Then $L_J = d(J, J')$ is the length of the gray edge between $J$ and $J'$.  \label{length-gray}
            \item  for any leaf-to-root path of $\tau$, the $L$-values of nodes are non-increasing. 
            \label{lem:leaf-to-root-original-gray}
		\end{enumerate}	
     \end{restatable}

	\subsection{Capacitated Clustering with Penalty Problem}\label{subec:Penalty-Problem}  We now introduce the key contribution  discussed in Section~\ref{sec:techniques}: the use of capacitated clustering with penalty instances to generate valid inequalities on the original capacitated $k$-median instance.

	Given a black component $J \in \cal{J}$, and an integer parameter $t \geq 0$, the capacitated clustering with penalty instance (or simply, penalty instance) $\calI_J(t)$ is defined as follows. We have demands at $J$, where each $v \in J$ has $x^*(U_v, C)$ units of demand. The facilities are $U(J) = \union_{v \in J}U_v$, and each facility $i \in U(J)$ has capacity $u_i$.   We are allowed to open at most $t$ facilities.

	Crucially, for the instance $\calI_J(t)$, we are allowed to not connect all the demands: for every unit-demand that is not connected, we pay a penalty cost of $L_J / 2$.  Our goal is to open at most $t$ facilities $S$ in $U(J)$, and connect some demands from $J$ to $S$, so that every $i \in S$ receives at most $u_i$ units of demand. Connecting demands at $v \in J$ to $i \in U(J)$ incurs a per-unit cost of $d(v, i)$. For all the demands that are not connected to $S$, we pay a per-unit penalty cost of $L_J / 2$. The goal of the problem is to minimize the sum of connection and penalty costs.

	\paragraph{$O(|J|)$-Approximate Solution $S_J(t)$ for $\calI_J(t)$} As $|J|$ is small, we can afford to use an $O(|J|)$-approximation for the problem.\footnote{Alternatively, we could solve the problem exactly with running time exponential in $|J|$. However, it is more important to achieve $\poly(n, 1/\epsilon)$ running time, and we lose approximation factors of $\poly(1/\epsilon)$ anyway. So we use the $O(|J|)$-approximation in $n^{O(1)}$-time.} 
We first convert the metric $d$ to a hierarchical tree metric. Let $m:=|J|$ and $\alpha:=2m-1$. Let $T$ consist of the black MST edges over $J$, together with an edge $(v,i)$ for every $v\in J$ and $i\in U_v$. Representatives and facilities are treated as separate vertices, even if they are at the same location. Assign auxiliary weights
\begin{align*}
    \lambda(v,w)&:=d(v,w)/2 &&\text{for every black edge }(v,w),\\
    \lambda(v,i)&:=d(v,i) &&\text{for every }i\in U_v.
\end{align*}
For distinct vertices $a,b$ of $T$, let $b_T(a,b)$ be the maximum auxiliary weight on their unique path in $T$, and set $b_T(a,a):=0$. Define $d'(a,b):=\alpha b_T(a,b)$. This is an ultrametric, and the Kruskal decomposition described below realizes it as a tree metric with all original vertices at leaves. The different scaling of black edges and facility edges allows us to retain the factor $2|J|-1$.

\begin{lemma}\label{lem:tree-metric}
		For every $v \in J$ and $i \in U(J)$, we have $d(v, i) \leq d'(v, i) \leq (2|J| - 1) d(v, i)$.
	\end{lemma}
    \begin{proof}
Fix $v\in J$ and $i\in U_w$ for some $w\in J$, and let $b:=b_T(v,i)$. The path from $v$ to $i$ in $T$ has at most $m-1$ black edges, each of original length at most $2b$, and one facility edge, of length at most $b$. By the triangle inequality, $d(v,i)\leq(2m-1)b=d'(v,i)$.
For the reverse inequality, $d(w,i)\leq d(v,i)$ since $w$ is a closest representative to $i$. Every black edge on the MST path from $v$ to $w$ has length at most $d(v,w)\leq d(v,i)+d(i,w)\leq2d(v,i)$. Thus every auxiliary weight on the path from $v$ to $i$ is at most $d(v,i)$, giving $d'(v,i)\leq\alpha d(v,i)$.
    \end{proof}

    By losing a factor of $(2|J|-1)$ on the connection cost, we can consider the tree metric $d'$.  Furthermore, we can make the following two assumptions on the solution for the instance $\calI_J(t)$:
    \begin{enumerate}[label=(D\arabic*)]
        \item First, if some $i \in U(J)$ has $d(i, J) > L_J/2$ (notice that $d(i, J) = d(i, v)$ for the $v \in J$ with $i \in U_v$), then $i \notin S_J(t)$. This can be assumed since such $i$ will have distance more than $L_J/2$ to all demands and thus it would be better to pay the penalty for the demands than to connect them to $i$.
        \item 
Second, either all the capacities of the open facilities are used, or all the demands are connected. After (D1), every auxiliary edge on the path from any $v\in J$ to any remaining facility $i$ has weight at most $L_J/2$: this follows from Claim~\ref{claim:componentJ} for black edges, and from (D1) for the facility edge. Hence $d'(v,i)\leq\alpha L_J/2$. Whenever both unconnected demand and unused open capacity remain, we can connect them at a per-unit cost of at most $\alpha L_J/2$ instead of paying a penalty of $L_J/2$. Starting with an optimum solution in the original metric, its original connection costs increase by at most $\alpha$, and each originally penalized unit costs at most $\alpha L_J/2$ after this modification, whether it is connected or left unconnected. Thus (D1), (D2), and the metric conversion together lose only a factor of $\alpha$.
    \end{enumerate}

    As all capacities $u_i$ for $i \in U(J)$ are integers in $[n]$, the assumption in (D2) allows us to solve the modified instance precisely using dynamic programming.

    We give a more formal description of the dynamic programming now. For the sake of notations, it is convenient for us to introduce a standalone problem, give a dynamic programming for this problem, and we show how to reduce our problem to this standalone problem.

\medskip

In the standalone problem, we have a full binary tree $H=(V_H,E_H)$ rooted at $r_H$, with leaf set $\mathcal L_H$. Each vertex $v$ has a height $h_v\geq0$, where leaves have height zero and heights are non-decreasing towards the root. For a parent $v$ and child $w$, the edge length is $\omega(v,w):=h_v-h_w$. The resulting tree metric satisfies
\begin{align*}
    d_H(z,z')=2h_{\operatorname{lca}(z,z')}
    \qquad\text{for distinct leaves }z,z',
\end{align*}
where $\operatorname{lca}(z,z')$ denotes their lowest common ancestor. Each leaf $z\in\mathcal L_H$ has a capacity $u_z\in\{0,\ldots,u_{\max}\}$ and a demand $x_z\in\mathbb R_{\geq0}$. Only leaves may be opened as facilities; internal vertices have neither capacity nor demand. We are also given integers $\kappa, P\geq0$. For every $B\subseteq\mathcal L_H$, let
$u(B):=\sum_{z\in B}u_z$.

The goal is to open a set $V'\subseteq\mathcal L_H$ of at most $\kappa$ facilities with total capacity $u(V')=P$, and connect exactly $\min\{x(H),P\}$ units of demand at minimum cost, respecting the capacity of each open facility. Here $x(H)$ is the total leaf demand, and connections may be fractional. Equivalently, either all demands are connected or all open capacity is used. \medskip

The hierarchical metric ensures that there is an optimum assignment that matches as much demand as possible within every subtree before connecting to outside. More precisely, for any fixed open set $V'$ and every subtree $H_v$ rooted at $v$, we can match exactly $\min\{x(H_v),u(V'\cap\mathcal L_{H_v})\}$ units between leaves of $H_v$, where $\mathcal L_{H_v}$ is its leaf set. To see this, suppose this amount is not reached. Some demand in $H_v$ is unconnected or connected outside, and some open capacity in $H_v$ is unused or serves outside demand. Match a positive amount of this demand and capacity internally. If both were matched outside, reconnect the displaced outside demand and capacity to each other; otherwise, leave the displaced endpoint unmatched. The internal distance is at most $2h_v$, whereas the distance
between any leaf in $\mathcal L_{H_v}$ and any leaf in $\mathcal L_H\setminus\mathcal L_{H_v}$ is at least $2h_v$. Moreover, the distance between two outside leaves is at most the larger of their distances to $\mathcal L_{H_v}$, by the ultrametric inequality. Thus this exchange does not increase cost and preserves the total connected demand. The case of both endpoints being unmatched is impossible, since all demand or all open capacity is used. Applying these exchanges from the leaves upwards preserves matches already made within descendant subtrees and gives the claimed assignment.
\medskip

We give a $\poly(u_{\max},|V_H|)$-time dynamic programming algorithm to solve this problem exactly. For every $v\in V_H$, let $H_v$ be the subtree rooted at $v$, and let $x(H_v)$ be its total leaf demand. For integers $a\in\{0,\ldots,\kappa\}$ and $p\in\{0,\ldots,P\}$, define $\text{cost}(v,a,p)$ to be the minimum cost of opening exactly $a$ leaves in $H_v$ with total capacity $p$, and connecting $\min\{x(H_v),p\}$ units of demand entirely within $H_v$. Infeasible states have value $\infty$.

For a leaf $v$, the feasible states are $\text{cost}(v,0,0)=0$ and $\text{cost}(v,1,u_v)=0$, whenever their indices are in range. For an internal vertex $v$ with children $v_1,v_2$, write $X_j:=x(H_{v_j})$ for $j\in\{1,2\}$. For any split $a_1+a_2=a$ and $p_1+p_2=p$, define
\begin{align*}
    q_v(p_1,p_2):=\min\{X_1+X_2,p_1+p_2\}
       -\min\{X_1,p_1\}-\min\{X_2,p_2\}.
\end{align*}
This nonnegative quantity is exactly the amount that must be matched between the two child subtrees after each child has matched as much demand as possible internally. Each such connection costs $2h_v$ per unit, independently of its endpoints. Therefore,
\begin{align*}
    \text{cost}(v,a,p)=\min_{\substack{a_1+a_2=a\\p_1+p_2=p}}
    \Big\{&\text{cost}(v_1,a_1,p_1)+\text{cost}(v_2,a_2,p_2)
          +2h_v\,q_v(p_1,p_2)\Big\},
\end{align*}
where the minimum ranges over states with valid indices. The preceding exchange argument shows that an optimum assignment may maximize the internal matching in each child. Conversely, optimal child solutions can always be combined by matching $q_v(p_1,p_2)$ units of their residual demand and capacity. Any residual endpoints left for ancestors have the same distance to every fixed outside leaf, so their identities need not be included in the state. This proves the recurrence.

The optimum value is $\min_{0\leq a\leq\kappa}\text{cost}(r_H,a,P)$. We may cap $\kappa$ at $|\mathcal L_H|$ and declare the instance infeasible if $P>\sum_{z\in\mathcal L_H}u_z$. The table has $O(|V_H|(\kappa+1)(P+1))$ states and can be computed in $O(|V_H|(\kappa+1)^2(P+1)^2)$ arithmetic operations, which is polynomial in $u_{\max}$ and $|V_H|$. An optimal facility set and assignment are recovered by tracing the minimizing splits and matching residual demand and capacity at each merge. \medskip

We now reduce the hierarchical tree-metric problem restricted by (D1) and (D2) to the standalone problem. Let $A:=\{i\in U(J):d(i,J)\leq L_J/2\}$ be the remaining facilities, and delete the other facility vertices from $T$. Run Kruskal's algorithm on this tree using the auxiliary weights $\lambda$, breaking ties arbitrarily. Its decomposition tree $H$ has one leaf for each vertex in $J\cup A$. Initially, each singleton component is represented by its leaf. Whenever an edge $e$ merges two components, create a new parent whose two children are the nodes representing those components, and give this parent height $\alpha\lambda(e)/2$. Give all leaves height zero. Thus every node corresponds to a component that appeared during Kruskal's algorithm, and internal nodes have exactly two children. Equal merge weights simply give zero-length edges in $H$.

For any two leaves $a,b$, the merge weight at their lowest common ancestor is $b_T(a,b)$. Therefore $d_H(a,b)=\alpha b_T(a,b)=d'(a,b)$. The construction uses $2(|J|+|A|)-1$ vertices. Assign demand $x_v=x^*(U_v,C)$ and capacity zero to the leaf for each representative $v\in J$. Assign demand zero and capacity $u_i$ to the leaf for each facility $i\in A$; internal nodes have neither. Set $\kappa:=\min\{t,|A|\}$ and $u_{\max}:=\max(\{0\}\cup\{u_i:i\in A\})$. Zero-capacity leaves never need to be opened, so any standalone solution yields a set of at most $t$ original facilities.

Enumerate all integers $P\in\{0,\ldots,\sum_{i\in A}u_i\}$. For each $P$, let $S_H(P)$ denote the optimum connection cost of the standalone instance, or $\infty$ if it is infeasible. Let $X:=x^*(U(J),C)$. By (D2), exactly $\min\{X,P\}$ units are connected. Hence the optimum value of the modified problem is
\begin{align*}
    \min_{P\in\{0,\ldots,\sum_{i\in A}u_i\}}
    \left\{S_H(P)+\frac{L_J}{2}\bigl(X-\min\{X,P\}\bigr)\right\}.
\end{align*}
We select a minimizing value of $P$, recover its facility set and assignment, and use the same connections in the original metric $d$. By Lemma~\ref{lem:tree-metric}, this cannot increase cost. Since $u_i\leq n$ and $|A|\leq n$, the enumeration and dynamic programming take polynomial time. By the joint charging argument for the metric conversion and (D2), the returned solution is an $\alpha=(2|J|-1)$-approximation for $\mathcal I_J(t)$.

     Let $S_J(t)$ be the solution obtained using the dynamic programming algorithm, $o_J(t)$ be the amount of demand that is not connected in the solution (that is, we pay a penalty for $o_J(t)$ units of demand), and $\Psi_J(t)$ be the cost of the solution.  As we argued, $S_J(t)$ is a $(2|J|-1)$-approximate solution for the instance $\calI_J(t)$. So $\Psi_J(t)$ is at most $(2|J| - 1)$ times the optimum cost of $\calI_J(t)$.

\subsection{Generating Valid Constraint Using $S_J(t)$ for $\calI_J(t)$}
	
	For a given black component $J \in \calJ$ and parameter $t \geq 0$,  we can use the solution $S_J(t)$ for $\calI_J(t)$ to generate a valid constraint of form \eqref{LPC:penalty}.

	\begin{lemma} \label{valid-constraint}
		For any $J \in \calJ$  and integer $t \geq 0$,  the following is a valid constraint for the original CKM instance:
		\begin{align}
			\frac{\Psi_J(t)}{t+1} y(U(J)) + (2|J|-1) \sum_{v \in J, j \in C} g_{j, v} (\phi_j + d(v, j)) \geq \Psi_J(t). \label{equ:all-constraint-J-t}
		\end{align}
	\end{lemma}
	\begin{proof} 
        Consider an integral solution $(x, y, \phi)$ to LP\eqref{LP:basic}. If $y(U(J))\geq t+1$, then the inequality clearly holds as the first term is already at least $\Psi_J(t)$.
        
        So we assume $y(U(J))\leq  t$.  Assume towards the contradiction that $\frac{\Psi_J(t)}{t+1} y(U(J)) + (2|J|-1) \sum_{v \in J, j \in C} g_{j, v} (\phi_j + d(v, j)) < \Psi_J(t)$. In particular, $(2|J|-1) \sum_{v \in J, j \in C} g_{j, v} (\phi_j + d(v, j)) < \Psi_J(t)$. For every $j\in C$,  let $i_j$ be the open facility connected to $j$ in the integral solution. So $x_{i_j, j}=1$. 
          
      We then construct a solution to the capacitated clustering with penalty instance $\calI_J(t)$. We open the set of facilities $i \in U(J)$ with $y_i = 1$. So, we open at most $t$ facilities. For every $v \in J$, and every $j \in C$ with $g_{j, v} > 0$, we consider how to connect the $g_{j, v}$ units of demand moved from $j$ to $v$ at the initial moving phase:
       \begin{itemize}
            \item If $i_j\in U(J)$, we connect the demand to $i_j$. In this case, the connection cost for the demand is $g_{j, v} d(v,i_j) \leq g_{j, v}(d(v, j)+d(j, i_j))\leq g_{j, v}(\phi_j + d(v, j)).$ The first inequality holds by the triangle inequality, and the second by \eqref{LPC:phi-j}.
             \item If $i_j\notin U(J)$, we treat the demand as unconnected, and pay a penalty of $g_{j, v} L_J/2$. Notice that $d(v, i_j) \geq L_J/2 = d(J, R \setminus J)/2$ since otherwise $i_j$ would be in $U(J)$.  So, the penalty we pay is at most $g_{j, v} d(v, i_j) \leq g_{j, v} (\phi_j + d(v, j))$; again, we used the triangle inequality and \eqref{LPC:phi-j}.
        \end{itemize}
        So, the cost of the solution is at most $\sum_{v \in J, j \in C} g_{j, v} (\phi_j + d(v, j)) < \frac{\Psi_J(t)}{(2|J|-1)}$, which contradicts the fact that $S_J(t)$ is a $(2|J| - 1)$-approximate solution to $\calI_J(t)$.
    \end{proof}

    Notice that the inequality in Lemma~\ref{valid-constraint} can be written as an inequality of form \eqref{LPC:penalty}.  We can check if our $(x^*, y^*, \phi^*)$ satisfies the constraint for every $J \in \calJ$ and integer $t \in [0, 2|U(J)|]$. If not, we return a constraint that is violated.  So from now on, we assume $(x^*, y^*, \phi^*)$  satisfies all these constraints. As we shall see later, for our rounding algorithm to succeed, we only need a small subset of these constraints to hold, but we can afford to check all of them. 
    
    So far, we have finished the description of the algorithm that is common to both the cardinality violation and capacity violation settings.  We describe the two algorithms separately in the next two sections. As we mentioned, the algorithms follow a similar framework: we partition the black components into groups, construct local solutions for black components, and use an auxiliary LP to obtain a fractional solution with only a few fractional facility variables, and round these variables to integral with a $(1+\epsilon)$-violation. 
    \section{Algorithm for $(1+\epsilon)$-Cardinality-Violation Setting}
\label{sec:Cardinality-Violation}
In this section, we describe the algorithm for the cardinality-violation setting. 
 Our algorithm builds on the gray-forest grouping framework
of \cite{Li16}. We partition the black components into groups
and round each group independently. The main difference
lies in the construction and analysis of the local solutions.
We use the penalty solutions $S_J(t_J)$ from Section \ref{subec:Penalty-Problem}
to define the local opening vectors and the auxiliary
LP. Their costs are bounded using the valid
inequalities \eqref{equ:all-constraint-J-t}, rather than the configuration-LP
constraints used in \cite{Li16}.

\subsection{Partition Gray Forest into Groups of Weight $\Theta(\ell)$} 
We partition each gray tree $\tau$ in the gray forest into groups of weight $\Theta(\ell)$, using a bottom-up greedy approach similar to that in \cite{Li16}. (The only exception is that the root of $\tau$ may be a singleton group, with weight much larger than $\Theta(\ell)$.)   Focus on such a tree $\tau$. In each iteration, select the deepest node  (black component) $J$ in $\tau$ such that the total weight of all descendants of $J$ is at least $\ell$; recall that the weight of a black component $J'$ is defined as $y^*(U(J'))$. As the root of $\tau$ has weight at least $\ell$, such $J$ exists.

By Property \ref{le:weight-J}, every non-root node in $\tau$ has weight less than $\ell$. By the choice of $J$, the total weight of the subtree rooted at any child of $J$ is less than $\ell$.  If the total weight of all these subtrees is at most $\ell$, then we take $G$ to be the set of black components in the whole sub-tree rooted at $J$ and let $\mathrm{root}_G = J$. Otherwise, a greedy procedure can select a subset $\calJ'$ of children  of $J$ such that the total weight of subtrees rooted at $\calJ'$ is between $\ell$ and $2\ell$. In this case, we let $G$ contain the black components in these subtrees, and define $\mathrm{root}_G$ to be child $J' \in \calJ'$ with the smallest $L_{J'}$.  Notice that in this case $G$ is not connected in $\tau$, as we did not include $J$ in $G$.

In both cases, we remove $G$ from the tree $\tau$ and repeat until $\tau$ becomes empty. We define $\mathcal G$ as the set of all groups we constructed, over all trees $\tau$ in the gray forest.  Clearly, $\mathcal{G}$ is a partition of $\calJ$. The total weight of black components in any $G$ is at least $\ell$.

 \begin{restatable}{lemma}{lemmagroups}\label{lemmagroups}
	For every $G \in \calG$, every $J \in G$ and $v \in J$, the distance between $v$ and any representative in $\mathrm{root}_G$ is at most $O(\ell) \cdot L_J$.  \label{amount-un-contracting}
\end{restatable}

       \subsection{Construction of Local Solutions for Black Components $J \in \calJ$}
       
       	For every black component $J \in \calJ$, we define $t_J:= \lfloor (1+1/\ell)y^*(U(J))\rfloor$. We shall only use the capacitated clustering with penalty instance $\calI_J(t_J)$ for each $J$. For convenience, we use $\calI_J, S_J, o_J$ and $\Psi_J$ for $\calI_J(t_J), S_J(t_J), o_J(t_J)$ and $\Psi_J(t_J)$ respectively. 
       	
       Fix a black component $J \in \calJ$, and let $S_J$ be the $(2|J|-1)$-approximate solution to the instance $\calI_J$ with cost  $\Psi_J$. In this solution, we construct a local solution $\bar y \in [0, 1]^{U(J)}$ for this $J$, along with a per-unit cost vector $b \in \R_{\geq 0}^{U(J)}$.

       For every $i \in U(J)$, we define $u'_i := \frac{x^*(i, C)}{y^*_i} \leq u_i$ to be the ``effective capacity'' of $i$, which is the capacity used by the LP solution $(x^*, y^*, \phi^*)$. $J$ is one of the following two types:
       \begin{enumerate}[label=(\arabic*)]        
   		\item if $\displaystyle {o_J}\leq \frac{1}{\ell}\cdot x^*(U(J) \setminus S_J, C)$, which is equivalent to $x^*(U(J) \setminus S_J, C) \geq \ell {o_J}$. 
           In this case, we let $\bar y_i = 1$ for every $i \in S_J$, and 
   		$\bar y_i = \frac{{o_J}}{x^*(U(J) \setminus S_J, C)} \cdot y^*_i \leq \frac1\ell \cdot y^*_i$ for every $i \in U(J) \setminus S_J$. That is, $\bar y$ is the indicator vector for $S_J$, plus $\frac{{o_J}}{x^*(U(J) \setminus S_J, C)} \leq \frac 1\ell$ times $y^*$ restricted to $U(J) \setminus S_J$.   We let $S'_J = S_J$ in this case.
           
   		\item  if $\displaystyle {o_J}>\frac{1}{\ell}\cdot x^*(U(J) \setminus S_J, C)$. We let $S'_J = \{i \in S_J: y^*_i \geq 1 - \frac1\ell\} \subseteq S_J$.
             Then, we define $\bar y_i = 1$ for every $i \in S'_J$ and $\bar y_i = y^*_i$ for every $i \in U(J) \setminus S'_J$.  So, $\bar y$ is $y^*$ with values for $S'_J$ scaled up to $1$. 
   	\end{enumerate}

    In both cases we have $S'_J = \{i \in S_J:\bar y_i =1\}$. Then for both cases, we define $b_i = 0$ for every $i \in S'_J$, and $b_i = d(i, J) + L_J$ for every $i \in U(J) \setminus S'_J$.

   	\begin{lemma}
       The following inequalities hold:
       \begin{enumerate}[label=(\thelemma\alph*),leftmargin=*]
           \item $\bar y(U(J)) \leq (1 + \frac2\ell) y^*(U(J))$. \label{facility-open} 
           \item $\sum_{i \in U(J)} u'_i b_i \bar y_i \leq O(1) \cdot D(U(J)) + O(\ell^2) \cdot \Psi_J$.\label{conec-cost}
   		\end{enumerate}			
   	\end{lemma}
   	\begin{proof}
       We first prove \ref{facility-open}. We consider the two cases (1) and (2) above. For type-(1) component $J$,   we have $\bar y(U(J)) \leq |S_J| + \frac1\ell y^*(U(J) \setminus S_J) \leq (1 + \frac{1}{\ell})y^*(U(J)) + \frac{1}{\ell} y^*(U(J)) = (1 + \frac2\ell)y^*(U(J))$. For type-(2) component $J$,
       $\bar y(U(J))=|S_J'|+y^*(U(J)\setminus S_J')\leq (1+\frac{1}{\ell-1})y^*(S_J')+y^*(U(J)\setminus S_J')=(1+\frac{1}{\ell-1})y^*(U(J))$. The first inequality holds by $y^*_i\geq 1-\frac{1}{\ell}$ for each $i\in S_J'$.

       It remains to consider \ref{conec-cost}.
       For type-(1) component $J$, we have
       \begin{align*}
           \sum_{i \in U(J)} u'_i b_i \bar y_i &=\frac{o_J} {x^*(U(J) \setminus S_J, C)} \cdot \sum_{i \in U(J) \setminus S_J}  u'_i y^*_i (d(i, J) +  L_J)\\
           &\leq \frac{1}{\ell}\cdot \sum_{i\in U(J)\setminus S_J}u'_i y_i^* d(i, J)+  {o_J} L_J \\
           &= \frac1\ell\sum_{i\in U(J) \setminus S_J}x^*_{i, C} d(i, J)+ o_J L_J\leq O(\frac1\ell)\cdot D(U(J)) + 2 \cdot \Psi_J.
       \end{align*}
       Recall that $u'_i y^*_i = x^*_{i, C}$ for every $i \in U(J) \setminus S_J$. The last inequality follows from Lemma~\ref{lemma:x*iCdiv} and $\Psi_J\geq  {o_J}\cdot \frac{L_J}{2}$.
       
    For type-(2) component $J$, we have 
           $\sum_{i \in U(J)}u'_i b_i \bar y_i = \sum_{i \in U(J)\setminus S'_J} u'_i y^*_i (d(i, J)+   L_J) \leq O(1) \cdot D(U(J)) +  x^*(U(J) \setminus S'_J, C) \cdot L_J$. We shall bound $x^*(U(J) \setminus S'_J, C)$. As $J$ is of type-(2), we have $o_J > 0$, and $x^*(U(J) \setminus S_J, C) < \ell o_J$. By our assumption on the solution $S_J$, all capacities in $S_J$ are fully used in the solution $S_J$ for $\calI_J$. Therefore, we have 
           \begin{align*}
               x^*(U(J) \setminus S_J, C) &\geq o_J + \sum_{i \in S_J \setminus S'_J} (1 - y^*_i) u_i \geq \sum_{i \in S_J \setminus S'_J} (1 - y^*_i) u'_i \geq \frac{1}{\ell-1} \sum_{i \in S_J \setminus S'_J} y^*_i u'_i \\
               &= \frac{1}{\ell-1} \cdot x^*(S_J \setminus S'_J, C). 
           \end{align*}
           
           So, we have $x^*(S_J \setminus S'_J, C) \leq (\ell-1) x^*(U(J) \setminus S_J, C) \leq \ell(\ell-1)o_J$. Overall, we have $x^*(U(J) \setminus S'_J, C) \leq \ell^2 o_J$. So, we have $\sum_{i \in U(J)} u'_i b_i \bar y_i \leq O(1)D(U(J)) + \ell^2 o_JL_J \leq O(1)D(U(J)) + O(\ell^2) \cdot \Psi_J$. This finishes the proof of \ref{conec-cost} and thus the whole lemma.
   	\end{proof}

\subsection{Rounding Procedure for Each Group $G \in \mathcal G$} Our rounding algorithm handles each group $G \in \mathcal G$ independently.  We focus on a fixed $G \in \mathcal G$. Recall that every black component $J \in G$, we defined two vectors $\bar y \in [0, 1]^{U(J)}$ and $b \in \R_{\geq 0}^{U(J)}$.  We then let $\bar y$ and $b$ be the concatenation of all black components $J \in G$. So, the domain of both $\bar y$ and $b$ is $U(G) := \union_{J \in G}U(J)$.

    Consider the following LP with variables $\{\tilde y_i\}_{i\in U(G)}$:
    \begin{align}
        \min \qquad  \sum_{i \in U(G)}  u'_i b_i \tilde y_i\label{LP:round-bar-y}
\end{align}\vspace*{-30pt}

\begin{align}
       \tilde y(U(G)) \leq \bar y(U(G));  \qquad \sum_{J \in G}\sum_{i \in U(J) \setminus S'_J}  u'_i \tilde y_i \geq \sum_{J \in G}\sum_{i \in U(J) \setminus S'_J}  u'_i \bar y_i; \nonumber \\
       \tilde y_i\in [0, 1], \forall i\in U(G); \qquad  \tilde y_i = 1, \text{ if } \bar y_i = 1.\nonumber
    \end{align}
    By setting $\tilde y_i = \bar y_i$,  we obtain a solution to LP(\ref{LP:round-bar-y}) of value
 $\sum_{i \in U(G)} \tilde y_i u'_i b_i \leq \sum_{i \in U(G)} \bar y_i u'_i b_i\leq \sum_{J\in G, i\in U(J)}(x({i, C}^*) d(i, J)+4\ell^2 \cdot \Psi_J)$.   Fix on such an optimum vertex-point solution $\tilde y$ of LP(\ref{LP:round-bar-y}). Since there are only two non-box-constraints, $\tilde y$ has at most two fractional values. We then open all facilities $i \in U(G)$ with $\tilde y_i > 0$.

     We now reassign demands to these finally opened facilities. We define the root representative of $G$ to be any representative in the root black component of $G$. The whole process consists of  the following  two steps.
    \begin{itemize}
        \item Sending step. For each black component $J\in G$, we consider the solution $S_J$ for the capacitated clustering with penalty instance $\calI_J$.  Recall that $S'_J  = S_J$ if $J$ is of type-(1) and $S'_J \subseteq S_J$ if $J$ is of type-(2). Then we pre-connect demands at $J$ to $S'_J$ according to the solution $S_J$.  For the demands that are not connected, we send them to the designated root representative of $G$.  Notice that every $i \in S'_J$ has $\tilde y_i = \bar y_i = 1$, and its full capacity $u_i \geq u'_i \geq x^*(i, C)$ is used in the local solution, unless $o_J = 0$. So, for a $J \in G$, the amount of demands sent from $J$ to the root is  at most $x^*(U(J) \setminus S'_J, C)$.
        \item Receiving step. For the demands sent to the root representative in the previous step, we send them to the open facilities in $\union_{J \in G}U(J)$. For every $J \in G$, only facilities in $U(J)\setminus S'_J$ are allowed to receive demand, and such a facility $i$ receives at most $u'_i\tilde y_i$ units of demand from the root representative.  This is the only requirement for the receiving step. 

        So, the total amount of demands sent to the root is at most 
        \begin{align*}
            \sum_{J \in G} \sum_{i \in U(J) \setminus S'_J}u'_i \bar y_i \leq \sum_{J \in G} \sum_{i \in U(J) \setminus S'_J}u'_i \tilde y_i.
        \end{align*}
        Therefore, we have enough capacity to receive the demands sent to the root representative.
    \end{itemize}

\begin{lemma}
    The connection cost incurred by the sending and receiving steps is at most $O(\ell) \cdot D(U(G)) + O(\ell^3) \cdot \sum_{J \in G} \Psi_J$.
\end{lemma}
  \begin{proof}
    Consider the sending step for each $J \in G$. The cost of local matching of demands according to the solution $S_J$ is at most $\Psi_J$.  Let $s_J$ be the amount of demand sent from $J$ to
the root representative after local pre-connection.
If $J$ is of type-(1), then $s_J=o_J$.
If $J$ is of type-(2), all facilities in $S'_J$ are
fully utilized, so $s_J\le x^*(U(J)\setminus S'_J,C)
\leq \ell^2o_J$. By Lemma~\ref{amount-un-contracting} and $L_Jo_J\le2\Psi_J$, the sending cost
is at most $O(\ell)L_Js_J\le O(\ell^3)\Psi_J.$
    The cost of $U(J) \setminus S'_J$  receiving demands from the root representative  is at most 
    \begin{align*}
        \sum_{i \in U(J) \setminus S'_J} \tilde y_i u'_i (O(\ell) \cdot L_J + d(i, J)) \leq O(\ell)\cdot \sum_{i \in U(J) \setminus S'_J} \tilde y_i u'_i b_i.
    \end{align*}

    Now we consider all black components $J \in G$. The cost of receiving the demands is at most 
    \begin{align*}
        &\quad O(\ell) \cdot \sum_{J \in G} \sum_{i \in U(J) \setminus S'_J} \tilde y_i u'_i b_i \leq O(\ell) \cdot \sum_{J \in G} \sum_{i \in U(J) \setminus S'_J} \bar y_i u'_i b_i \\
        &\leq O(\ell) \cdot D(U(G)) + O(\ell^3) \cdot \sum_{J \in G} \Psi_J.
    \end{align*}
    The first inequality follows from the optimality of
$\widetilde y$ for LP\eqref{LP:round-bar-y} and the feasibility of $\bar y$.
The second inequality follows from Lemma~\ref{conec-cost}.
   Summing the local pre-connection and sending costs over $J\in G$ gives at most $O(\ell^3)\sum_{J\in G}\Psi_J$.
  Combining this with the receiving-cost bound, the total
cost is at most $O(\ell)D(U(G))+O(\ell^3)\sum_{J\in G}\Psi_J.$
    \end{proof}

Then, we wrap up the analysis of cost by considering all groups $G$, as well as the cost incurred by initial moving. The total cost is at most
\begin{align*}
    O(\ell) D(F) + O(\ell^3)\cdot \sum_{J}\Psi_J \leq O(\ell^5) \sum_{j \in C}\phi^*_j.
\end{align*}
To see this, we just need to bound $\sum_J \Psi_J$. The inequality in Lemma~\ref{valid-constraint} holds for $(x^*, y^*, \phi^*)$, and every black component $J$. We add up the inequality over all $J$. We obtain 
\begin{align*}
    \sum_{J} \Psi_J &\leq O(\ell^2) \sum_J\sum_{v \in J, j \in C} g_{j, v}(\phi^*_j + d(v, j)) = O(\ell^2) \cdot \left(\sum_{j \in C} \phi^*_j  + (\text{cost of initial moving})\right) \\
    &\leq O(\ell^2) \cdot \sum_{j \in C}\phi^*_j.
\end{align*}
The equality holds because $\sum_{v}g_{j,v} = 1$ for every $j \in C$. Therefore, $\sum_{J}\Psi_J \leq O(\ell^2)\cdot \sum_{j \in C}\phi^*_j$. \medskip

The following lemma gives an upper bound on the number of open facilities:
\begin{lemma}
    The total number of open facilities is at most $(1+\frac{4}{\ell})\cdot 
    k$. 
\end{lemma}
\begin{proof}
    Focus on a fixed $G\in \mathcal G$, since the vector $\tilde y$ has at most two fractional entries, so $\lfloor\sum_{i\in U(G)}\tilde y_i\rfloor+ 2\leq \sum_{i\in U(G)}\tilde y_i+2\leq (1+\frac{2}{\ell})y^*(U(G))+2$.  Since every group has weight at least $\ell$, and the groups form a partition of all black components, we have  $\ell|\mathcal G|\leq \sum_{G\in\mathcal G}y^* (U(G))=\sum_{i\in F}y^*_i\leq k$. By  summing over all $\tau\in  \Upsilon$, $|\mathcal G|=\sum_{\tau }|\mathcal{ G}_\tau|\leq \frac{1}{\ell}\sum_{v\in R} y^*(U_v)=\frac{1}{\ell} \sum_{i\in F}y^*_i\leq \frac{k}{\ell}$.
    Then summing over all $G\in \mathcal G$, the total number of open facilities is at most $(1+\frac{2}{\ell})\sum_{G\in \mathcal G}y^*(U(G))+2|\mathcal G|\leq 
    (1+\frac{2}{\ell})k+\frac{2k}{\ell}=(1+\frac{4}{\ell})k$.
\end{proof}

With a large enough hidden constant in $\ell = \Theta(1/\epsilon)$, we finish the proof of Theorem~\ref{thm:main-card}.

    \section{Algorithm for the Capacity-Violation Setting}\label{sec:Capacity-Violation}

 This section presents the rounding algorithm for the capacity-violation setting.  In contrast to the cardinality-violation setting, we enforce the total number of opened facilities to be at most $k$, and each open facility $i$ is allowed to carry demand up to $(1+\epsilon)u_i$. 
 
We follow the binary-forest construction and the top-down
grouping framework of \cite{DemirciL16}.
We also use their general strategy of removing excess
facilities and reassigning demand to nearby groups.
Our local solutions are constructed differently: we use
the penalty instances $I_J(t)$ and the four component
types defined below. The auxiliary LP selects the
opening variables, and the cost analysis uses
inequalities~\eqref{equ:all-constraint-J-t}.

\subsection{Converting Gray  Forest $\Upsilon$ into A Forest $\Upsilon^*$ of Binary Trees}

The way we partition each gray tree into groups is different from the one we used for the cardinality-violation setting. In the capacity violation setting, we need to  use the tree structure over the groups, and it is important to make sure that the $L_J$ values of black components in a parent group is at most the $L_J$ values of black components in a child group.  Therefore, we use a top-to-bottom approach to construct the groups. Moreover, to bound the number of child groups of each group, we first make each gray tree binary using the left-child-right-sibling method. 

We use the same framework as in \cite{DemirciL16} to convert $\Upsilon$ into a binary forest. For every rooted tree $\tau \in \Upsilon$, we convert it to a tree $\tau^*$ with the same set of nodes (black components) as follows. For every non-leaf black component $J \in \tau$, we sort its children in non-decreasing order of $L$-values, and let the order be $J_1, J_2, \cdots, J_a$. Then, we add edges $(J_1, J), (J_2, J_1), (J_3, J_2), \cdots, (J_a, J_{a-1})$ to $\tau^*$, with the first node of each edge being the child, and the second node being the parent.  So, $\tau^*$ is a rooted tree with edges directed towards the root and every node in $\tau^*$ has at most $2$ children. 

	We call the tree $\tau^*$ a \emph{transformed gray tree}. 
	Let $\Upsilon^*$ be the set of all transformed gray trees $\tau^*$, over all gray trees $\tau \in \Upsilon$, and we call $\Upsilon^*$ the \emph{transformed gray forest}.

     \begin{restatable}{lemma}{lemmablackcomponent}
     \label{lemmablackcomponent}
         $\Upsilon^*$ satisfies the following properties:
         \begin{enumerate}[label=(\thelemma\alph*),leftmargin=*]
      \item   for any leaf-to-root path of $\tau^*\in \Upsilon^*$, the $L$-values of nodes are non-increasing.\label{lem:leaf-to-root-path}
             \item  every root component  $J\in \mathcal J$   of $\Upsilon^*$ has either $y^*({U(J)})\leq 2\ell$ or $|J|=1$.\label{le:|J|=1}
           \item for any non-root component $J$ and its parent $J'$ in $\Upsilon^*$, we have $\max_{v \in J, v' \in J'}d(v, v')\leq O(\ell) \cdot L_J$.\label{lem:the-distance-of-J-and-its-parent} 
		\end{enumerate}	
    \end{restatable}

  \subsection{Partition $\calJ$ into Groups using $\Upsilon^*$}\label{partion-group}
  	With the transformed gray forest $\Upsilon^*$, we can now partition the black components $\calJ$ into groups. We also use the framework of  \cite{DemirciL16}.  Let $\mathcal{G}= \emptyset$ initially. For every transformed gray tree $\tau^*$ in  $\Upsilon^*$ and its root $J$, we call construct-groups$(\tau^*, J)$, defined in Algorithm~\ref{alg:construct-groups}. 
      \begin{algorithm}[h]
      		\caption{construct-groups$(\tau^*, J)$}
      		\label{alg:construct-groups}
	      \begin{algorithmic}[1]
		      \State $G = \{J\}$ 
		      \While{$\sum_{J \in G} y^*(U(J)) < \ell$ and $G$ does not contain all descendants of $J$ in $\tau^*$}
		      	\State choose an edge $(J'', J')$ in $\tau^*$ with the smallest $L_{J''}$ satisfying $J' \in G, J'' \notin G$ (recall that edges are directed towards the root, so $J'$ is the parent of $J''$)
		      	\State $G \gets G \cup \{J''\}$
		      \EndWhile
		      \State $\mathcal{G}\gets \mathcal{G}\cup \{G\}$ 
		      \For{every edge $(J'', J')$ in $\tau^*$ such that $J' \in G, J'' \notin G$} \Comment{$J'$ is the parent of $J''$} 
                \State construct-groups$(\tau^*, J'')$ 
              \EndFor
	      \end{algorithmic}
      \end{algorithm}

    Notice that every group $G$ induces a (weakly) connected component in $\Upsilon^*$. Thus, every $G$ has a unique root black component. The forest structure naturally defines a parent-child relationship among groups: a group $G$ is the parent of a group $G'$ if the root of $G'$ is a child of some black component in $G$.   Also, the root of any transformed gray tree forms a singleton group, as its weight is at least $\ell$.  We call this a root group. 

    We extend the definition of $U(\cdot)$ to sets of black components, and in particular groups. For any set $\calJ' \subseteq \calJ$ of components, we define $U(\calJ'):= \union_{J \in \calJ'}U(J) = \union_{J \in \calJ', v \in J} U_v$.
 Let $\children(G)$ be the set of all child-groups of a group $G \in \calG$.  Let $\children'(G) = \children(G) \cup \{G\}$ if $G$ is  a root group, and $\children'(G) = \children(G)$ otherwise. 
 A root group $G\in\mathcal G$ is called an {isolated root group} if it has no child group, that is, $\children(G)=\emptyset$, then $\children'(G)=\{G\}$.

        \begin{restatable}{lemma}{lemmaGproperties}\label{lemmaGproperties}
         $\mathcal{G}$  satisfies the following properties:
         \begin{enumerate}[label=(\thelemma\alph*),leftmargin=*]
            \item if $G\in {\calG}$ is not a root group, then $y^*(U(G))<2\ell$;\label{lem-G:not a root group}
            
            \item  if $G\in {\calG}$ is a non-leaf group, then $y^*(U(G))\geq \ell$;\label{lemG: a non-leaf group}
            \item If $G\in \calG$ is an isolated root group, then $G=\{J\}$ and $y^*(U(J))\ge \ell$.\label{lem-G:isolated-root}
         
            \item any group $G$ has at most $O(\ell)$ child groups; \label{lem-G:the-amount-children-of-G}
            \item let $G$ be a group, $G' \in \children'(G)$, and $J' \in G'$, then $\max_{v' \in J', v \in J \in G}d(v, v') \leq O(\ell^2) L_{J'}$.\label{le:distant-G-G'}
		\end{enumerate}	
        \end{restatable}

  Observe that a set $J$ in an isolated root group can satisfy either $|J|=1$ or $|J|>1$.

  \subsection{Types of Black Components and Construction of Local Solutions }
  
  For every black component $J \in \calJ$, we define $t_J=\left\lfloor \left(1+\frac{1}{c \cdot \ell^2}\right) y^*(U(J))\right\rfloor$, for  a sufficiently large absolute constant $c> 0$. If $J$ is a root black component of some gray tree (equivalent, transformed gray tree), and $t_J > \ceil{y^*(U(J))}$, then we update $t_J = \ceil{y^*(U(J))}$.  Notice that for a non-root black component $J$, we have $t_J=\left\lfloor \left(1+\frac{1}{c \cdot \ell^2}\right) y^*(U(J))\right\rfloor \leq \ceil{y^*(U(J))}$ as $y^*(U(J)) < \ell$.

  Fix a  black component  $J\in \mathcal J$.  We consider the solution $S_J(t_J)$ for the penalty instance $\calI_J(t_J)$. When $t_J = \lceil y^*(U(J))\rceil > y^*(U(J))$, we may need to also consider the penalty instance $\calI_J(t_J-1)$.  Recall that $o_J(t_J)$ denotes the amount of unconnected demands in the solution $S_J(t_J)$, and $\Psi_J(t_J)$ is the overall cost of the solution $S_J(t_J)$. We assume $\Psi_J(t_J - 1) \geq \Psi_J(t_J)$; otherwise, we can use the solution $S_J(t_J - 1)$ for the instance $\calI_J(t_J)$.

  A  black component $J\in\mathcal J$ is one of four types. If $J$ is singleton root black component with $|J|=1$, then it is of type-(1). Otherwise, its type is:
    \begin{enumerate}[label=(\arabic*), 
    align=left, 
    leftmargin=*]
        \item  if $o_J(t_J) > \frac{1}{\ell}\cdot x^*(U(J), C)$,
            \item if $o_J(t_J) \leq \frac{1}{\ell}\cdot x^*(U(J), C)$ and $t_J \leq y^*(U(J))$, 
            \item if $o_J(t_J) \leq \frac{1}{\ell}\cdot x^*(U(J), C)$, $t_J > y^*(U(J))$ and $o_J(t_J-1) > \frac{1}{\ell}\cdot x^*(U(J), C)$, and    
        \item if $o_J(t_J) \leq \frac{1}{\ell}\cdot x^*(U(J), C)$, $t_J > y^*(U(J))$ and $o_J(t_J-1) \leq \frac{1}{\ell}\cdot  x^*(U(J), C)$.  
	\end{enumerate}

As before, we shall construct local solutions for the 4 types of components respectively. If $J$ is of type-(1), we simply use the fractional solution  $y^*|_{U(J)}$. If $J$ is of type-(2), we use the solution $S_J(t_J)$ for the penalty instance $\calI_J(t_J)$, with unconnected demands sent to $S_J(t_J)$, using \emph{proportional distribution}: the amount of demand sent to any $i \in S_J(t_J)$ is proportional to the amount of demand it already serves. In this case, we say we \emph{follow} the instance $\calI_J(t_J)$ for $J$.  If $J$ is of type-(4), we follow either $\calI_J(t_J - 1)$ or $\calI_J(t_J)$, depending on how many facilities will be open in $U(J)$ in the end.  If $J$ is of type-(3), we follow the instance $\calI_J(t_J)$ initially. If in the end we can only open $t_J - 1$ facilities in $U(J)$, we perform a remove operation to decrease the number of open facilities by 1.
 
\subsection{Construction of Final Solution Using Auxiliary Linear Program}

With the ideas of local solutions, we form  an auxiliary linear program to construct the final integral solution for our CKM instance.  Define $\calJ^{(1)}, \calJ^{(2)}, \calJ^{(3)}$ and $\calJ^{(4)}$ as the sets of type-(1), (2), (3) and (4) black components respectively. We shall use $\calJ^{(3, 4)}$ to denote $\calJ^{(3)} \cup \calJ^{(4)}$; the notations for other subsets of $\{1, 2, 3, 4\}$ are defined analogously.  And we define $\mathcal J_0$ to be the set of all singleton root black components with $|J|=1$. So $\calJ_0\subseteq \calJ^{(1)}$.
This section mainly focuses on the algorithm, and we defer the analysis of the cost and capacity violation to the next section. 
 
 For every $J \in \calJ^{(1)}$, and $i \in U(J)$, we have a variable $\tilde y_i \in [0, 1]$. If $\tilde y_i \in \{0, 1\}$, it indicates if $i$ is open or not in the final solution. For every $J \in  \calJ^{(3, 4)}$, we have a single variable $\tilde y_J \in [t_J - 1, t_J] \ni y^*(U(J))$. If $\tilde y_J \in \{t_J - 1, t_J\}$, then it denotes the number of facilities we open in $U(J)$.  For black components of $J \in \calJ^{(2)}$, as we always follow $S_J(t_J)$. So we do not have variables for them.

As in the previous section, we use $u'_i := \frac{x^*(i, C)}{y^*_i} \leq u_i$ as  the ``effective capacity'' of a facility $i \in F$.  The auxiliary LP is defined as follows:\medskip
\begin{align}\label{LP:auxiliary}
  \min   &\sum_{J\in \calJ^{(1)}, i\in U(J)} u'_i d(i, J)\tilde y_i + \sum_{J\in \calJ^{(1)}\setminus \calJ_0, i\in U(J)}O(\ell^2) u_i'L_J \tilde y_i  + \sum_{J \in\calJ^{(3, 4)}} O(\ell^3) (t_J - \tilde y_J)\Psi_J(t_J - 1) 
\end{align}\vspace*{-30pt}

\begin{align}
    \tilde y\big(U(\calJ^{(1)}\cap G')\big) - \floor{y^*\big(U(\calJ^{(1)}\cap G')\big)} &\geq 0, &\qquad&\forall G' \in {\calG}\label{LPC:auxiliary-y-(1)}\\
    \sum_{G' \in \children'(G), J \in \calJ^{(3, 4)}\cap G'} (t_J - {\tilde y}_J) &\leq 1,  &\qquad&\forall  ~\text{non-leaf}~G\in {\calG}\label{LPC:auxiliary-y-(3,4)}\\
    \sum_{J \in \calJ^{(1)}, i\in U(J)}{\tilde y_i}  + \sum_{J \in \calJ^{(2)}} y^*(U(J)) + \sum_{J \in \calJ^{(3, 4)}}{\tilde y}_J &\leq k
    , \label{LPC:auxiliary-k}\\
    \sum_{G' \in \children'(G), i \in U(\calJ^{(1)} \cap G')} \big({\tilde y}_i - y^*_i\big)u'_i &\geq 0,  &\qquad &\forall~\text{non-leaf}~G\in {\calG}\label{LPC:auxiliary-demand-type(1)}\\
    \sum_{i \in U(\calJ^{(1)} \cap G)} \big({\tilde y}_i - y^*_i\big)u'_i &\geq 0,  &\qquad &\forall~\text{isolated root}~G\in {\calG}\label{LPC:auxiliary-demand-type(1)-root}\\
    {\tilde y}_i &\in [0, 1], &\qquad &\forall J \in \calJ^{(1)}, i \in U(J)\label{LPC:auxiliary-yi-nn}\\
    {\tilde y}_J &\in [t_J-1, t_J], &\qquad &\forall J \in \calJ^{(3, 4)}\label{LPC:auxiliary-yJ-nn}
\end{align}

In the objective, the functions in $O(\ell^2)$ and $O(\ell^3)$ notations will be specified later. For an isolated root group $G=\{J\}$, constraint \eqref{LPC:auxiliary-y-(3,4)} is redundant: its left-hand side is either zero or $t_J-\widetilde y_J$, which is at most 1 by \eqref{LPC:auxiliary-yJ-nn}.

\begin{lemma}\label{lem:feasible- LPauxiliary}
	Setting $\tilde y_i= y^*_i$ for every $i\in U(J), J\in \calJ^{(1)}$ and  $\tilde y_J= y^*(U(J))$ for every $J \in \calJ^{(3)}\cup \calJ^{(4)}$ gives a valid solution to LP\eqref{LP:auxiliary}.
\end{lemma}

\begin{proof}
\eqref{LPC:auxiliary-y-(1)}, \eqref{LPC:auxiliary-demand-type(1)}, \eqref{LPC:auxiliary-demand-type(1)-root}, \eqref{LPC:auxiliary-yi-nn} and  \eqref{LPC:auxiliary-yJ-nn} hold trivially.

	Now consider \eqref{LPC:auxiliary-y-(3,4)}. 
	If $J$ is of type-(3) or (4), then $t_J = \floor{(1+\frac{1}{c\ell^2})y^*(U(J))} =\ceil{y^*(U(J))} > y^*(U(J))$. The equality holds as $y^*(U(J)) < 2\ell$ and thus $\frac{1}{c\ell^2} \cdot y^*(U(J)) < 1$.  Therefore, $t_J\leq (1+\frac{1}{c\ell^2})y^*(U(J))$. Therefore 
	\begin{align*}
		0 < t_J - y^*(U(J)) \leq \frac{1}{c\ell^2} \cdot y^*(U(J)) < 1.
	\end{align*}
	Summing up the inequality over all $J \in \calJ^{(3, 4)} \cap G'$ for a fixed $G' \in \children(G)$, we have 
	\begin{align*}
		\sum_{J \in \calJ^{(3, 4)} \cap G'} (t_J - y^*(U(J))) \leq \frac{1}{c\ell^2}\cdot \sum_{J \in \calJ^{(3, 4)} \cap G'}y^*(U(J)) \leq \frac{2}{c\ell}. 
	\end{align*}
    If $G'$ is a root group, then $G'=\{J\}$. If $J\notin\mathcal J^{(3,4)}$, the corresponding sum is 0. Otherwise, $J\notin\mathcal J_0$, so $|J|>1$, and Property \ref{le:|J|=1} implies $y^*(U(J))\le2\ell$.
	
    As $G$ has at most $O(\ell)$ children, we have 
	\begin{align*}
		\sum_{G' \in \children'(G), J \in \calJ^{(3, 4)} \cap G'} (t_J - y^*(U(J))) \leq 1,
	\end{align*}
	when $c$ is large enough.  Therefore, \eqref{LPC:auxiliary-y-(3,4)} holds.

	For every $J\in \mathcal J^{(2)}$, we have  $t_J = \lfloor y^*(U(J))\rfloor \leq y^*(U(J))$ facilities. Therefore, setting $\tilde y$ as in the lemma, we have $\sum_{J \in \calJ^{(1)}} \tilde y(U(J)) + \sum_{J \in \calJ^{(2)}}t_J + \sum_{J \in \calJ^{(3, 4)}}\tilde y_J \leq \sum_{J \in \calJ} y^*(U(J)) = y^*(F) \leq k$. Thus, \eqref{LPC:auxiliary-k} holds.
\end{proof}

We solve the LP to obtain a vertex solution $\tilde y$. From now on, we fix this $\tilde y$. 
\begin{lemma}\label{lem:vertex-solution}
	The following statements hold:
	\begin{enumerate}[label=(\thelemma\alph*),leftmargin=*]
		\item there are at most three facilities $i \in \union_{G' \in \children'(G)}U(\calJ^{(1)}\cap G')$ with $\tilde y_i \in (0, 1)$, for every non-leaf or isolated root group $G \in \calG$. \label{lem:vertex-solution-1}
		\item at most one black component $J \in \calJ^{(3, 4)}$ has $\tilde y_J \in (t_J - 1, t_J)$. \label{lem:vertex-solution-2}
	\end{enumerate}
\end{lemma}
\begin{proof}
	Consider the first statement \ref{lem:vertex-solution-1}. We fix a non-leaf or isolated root group $G \in \calG$. We say $i$ is a strictly fractional facility if $\tilde y_i \in (0, 1)$ (and $\tilde y_i$ exists). We say $G' \in \children'(G)$ contains a strictly fractional facility $i$, if additionally $i \in \calJ^{(1)} \cap U(G')$. 

    Suppose some group $G' \in \children'(G)$ contains at least 3 strictly fractional facilities $i_1, i_2, i_3 \in U(\calJ^{(1)} \cap G')$. Notice that there are at most three tight constraints involving one of the three variables $\tilde y_{i_1}, \tilde y_{i_2}, \tilde y_{i_3}$: \eqref{LPC:auxiliary-y-(1)} for $G'$, \eqref{LPC:auxiliary-k} and \eqref{LPC:auxiliary-demand-type(1)} for $G$. The coefficient matrix between the three constraints and three variables is $\begin{pmatrix}
	    1 & 1 & 1\\
        1 & 1 & 1\\
        u'_{i_1} & u'_{i_2} & u'_{i_3}
	\end{pmatrix}$.
    The rank of the matrix is less than $3$, contradicting that $\tilde y$ is a vertex solution. 
    
    Now suppose there are  two different groups in $\children'(G)$, each of which contains at least 2 strictly fractional facilities. There are at least 4 of these facilities, but the rank of the coefficient matrix defined as above is at most 3. Again a contradiction. 

    Now suppose there are three groups in $\children'(G)$, each of which contains exactly 1 strictly fractional facility. Notice that \eqref{LPC:auxiliary-y-(1)} is not tight for $G', G''$ and $G'''$. There are at most 2 tight constraints involving the one of the fractional variables. Again, $\tilde y$ would not be a vertex solution. 

    Finally, if there are three groups in $\children'(G)$, two of which contain exactly 1 strictly fractional facility, and one of which contains exactly 2 strictly fractional facilities, using the above method one can argue that the rank of the coefficient matrix is at most $3$.  Therefore, we conclude that there are at most $3$ strictly fractional facilities over all groups in $\children'(G)$.

	Consider the second statement \ref{lem:vertex-solution-2}. Fix a group  $G$, similarly, suppose there are at least 2 such components $J_1$ and $J_2$ involved in the same constraint from \eqref{LPC:auxiliary-y-(3,4)}. Then \eqref{LPC:auxiliary-y-(3,4)} and \eqref{LPC:auxiliary-k} might be tight for $G$. The coefficient matrix is $\begin{pmatrix}
	    -1 & -1 \\ 1 & 1
	\end{pmatrix}$, whose rank is 1. This contradicts that $\tilde y$ is a vertex solution.

   Now suppose there are two strictly fractional components in different groups. By integrality and the preceding argument, the corresponding constraints in \eqref{LPC:auxiliary-y-(3,4)}, if any, cannot be tight. Thus, only \eqref{LPC:auxiliary-k} can be tight for these variables, whose  rank at most 1, a contradiction.
\end{proof}

We shall maintain a set $W$ of open facilities, a flow from $R$ to $W$. 
Every $v \in R$ sends $x^*({U_v, C}) = \sum_{i \in U_v}u'_iy^*_i$ units of demand (this is the demand received by $v$ during the initial moving procedure in Section \ref{Initial-Moving-of-Demands}). At the beginning, $W$ may have size more than $k$. We shall use the ``remove'' procedure to remove facilities from $W$ to guarantee the upper bound. 

We start with $W = \emptyset$, and add facilities as follows:
\begin{itemize}
	\item For every $J \in \calJ^{(1)}$, we add $\{i \in U(J): \tilde y_i > 0\}$ to $W$. For every non-leaf or isolated root group $G \in \calG$, we send the demands at $\union_{G' \in \children'(G), J \in \calJ^{(1)} \cap G'} J$ to $W \cap \union_{G' \in \children'(G)} U(\calJ^{(1)} \cap G')$, where every $i$ receives  at most $u'_i\tilde y_i$ units of demand. The total demand and the total capacity are respectively 
    \begin{align*}
        \sum_{G' \in \children'(G)}\sum_{J \in \calJ^{(1)} \cap G'}\sum_{i \in U(J)} u'_i y^*_i \quad \text{and} \quad \sum_{G' \in \children'(G)}\sum_{J \in \calJ^{(1)} \cap G'}\sum_{i \in U(J)} u'_i\tilde y_i.
    \end{align*}
  By \eqref{LPC:auxiliary-demand-type(1)} for a non-leaf group, or \eqref{LPC:auxiliary-demand-type(1)-root} for an isolated root group,  the total demand is at most the total  capacity.

	\item For $J \in \calJ^{(2)}$, we add $S_J(t_J)$ to $W$. We satisfy the demands at $J$ using $S_J(t_J)$, by following the instance $\calI_J(t_J)$. To address a small technicality issue, for every $G' \in \calG$, we arbitrarily choose $\floor{\sum_{J \in \calJ^{(2)} \cap G'}(y^*(U(J)) - |S_J(t_J)|)}$ facilities in $\union_{J \in \calJ^{(2)} \cap G'}(U(J) \setminus S_J(t_J))$, and add them to $W$. They do not serve any demand and we call them ''dummy facilities''. 

    \item For $J \in \calJ^{(3)}$, we add $S_J(t_J)$ to $W$. We satisfy the demands at $J$ using $S_J(t_J)$, by following the instance $\calI_J(t_J)$.
	
	\item For $J \in \calJ^{(4)}$, we add $S_J(\ceil{\tilde y_J})$ to $W$. Similarly, we satisfy the demands at $J$ using $S_J(\ceil{\tilde y_J})$, by following the instance $\calI_J(\ceil{\tilde y_J})$.
\end{itemize}
 
For every $J\in\mathcal J^{(3,4)}$, let $q_J=t_J$ if
$J\in\mathcal J^{(3)}$, and let $q_J=\lceil\tilde y_J\rceil$
otherwise. After adding $S_J(q_J)$, we add
$q_J-|S_J(q_J)|$ arbitrary facilities from
$U(J)\setminus S_J(q_J)$ to $W$ as dummy facilities serving no demand. This is possible since
$q_J\le t_J\le\lceil y^*(U(J))\rceil\le |U(J)|$.
The local assignment, $o_J(q_J)$, and $\Psi_J(q_J)$
remain unchanged.
The size of $W$ may be more than $k$ for two reasons. First, some type-(1) component $J$ may have some facility $i \in U(J)$ with $\tilde y_i \in (0, 1)$ and we included it in $W$. Second, some type-(3) component may have $\tilde y_J = t_J - 1$ but we opened $t_J$ facilities in $U(J)$, including dummy facilities. To reduce the size of $W$, we run Algorithm~\ref{alg:reduce-W}.

The top-to-bottom order in Algorithm~\ref{alg:reduce-W} is chosen so that all black components of a group are processed before
any black component of its child groups.

\begin{algorithm}
	\caption{reduce the size of $W$ to at most $k$}
	\label{alg:reduce-W}
	\begin{algorithmic}[1]
		\For{every black component $J \in \calJ$ respecting the top-to-bottom order}
            \State let $G$ be the parent group of the group containing $J$, or $G = \{J\}$ if $J$ is a root group
			\If{$J$ is of type-(1)}
                \State let $a$ be the number of facilities $i \in W \cap U(J)$ with $\tilde y_i < 1$
                \State repeat $a$ times: remove$(G, J)$
             \ElsIf{$J$ is of type-(3) and $\tilde y_J = t_J - 1$}
                \State remove$(G, J)$
             \EndIf
		\EndFor
	\end{algorithmic}
\end{algorithm}


\begin{algorithm}
	\caption{remove$(G, J)$}
	\begin{algorithmic}[1]
        \State let $i^*$ be the facility in $W \cap U(G \cup \{J\})$ serving the least demand
		\State remove $i^*$ from $W$
		\State resend the demands sent to $i^*$ to $W \cap U(G)$, split proportionally to the amount of demands they already serve
	\end{algorithmic}
\end{algorithm}

We first argue that the algorithm is well-defined.  The initial $W$ we constructed has $|W \cap U(G')| \geq y^*(U(G')) - O(1)$ for every group $G'$.  This holds due to \eqref{LPC:auxiliary-y-(1)}, \eqref{LPC:auxiliary-y-(3,4)}, which implies $\sum_{J \in \calJ^{(3, 4)}\cap G'}(t_J - \tilde y_J) \leq 1$ for every group $G'$, and that we added dummy facilities for type-(2), type-(3) and type-(4) components.
If $G'$ is an isolated root group, the inequality follows directly from \eqref{LPC:auxiliary-yJ-nn}.
In particular, for every non-leaf or isolated root group $G \in \calG$, we have $|W \cap U(G)| \geq \Omega(\ell)$. By the first statement of Lemma~\ref{lem:vertex-solution}, for any non-leaf or isolated root $G \in \calG$, the number of times we call remove$(G, J)$ over all type-(1) components $J$ at most 3. 
For a non-leaf group, by \eqref{LPC:auxiliary-y-(3,4)}, the number of times we call remove$(G, J)$ over all type-(3) components $J$ at most 1. An isolated root group
contains only one black component, so the same bound
holds for it. Therefore, the facility $i^*$ can always be chosen. For every group $G$, at most four facilities in $U(G)$ are removed while processing components of $G$, and at most four more are removed while processing
its child groups. Hence every non-leaf or isolated root group $G$ retains $\Omega(\ell)$ open facilities
throughout the algorithm. 

The initial $|W|$ exceeds $k$ by at most the number of strictly fractional type-(1) facilities, plus the number of type-(3) components $J$ with $\tilde y_J = t_J - 1$. This holds due to the second statement of Lemma~\ref{lem:vertex-solution}.  As we removed this number of facilities, we have $|W| \leq k$ in the end.

\subsection{Analysis of Cost and Capacity Violation}
We analyze the cost and capacity violation of our solution. 

\paragraph{Cost of Initial Flow Before Removing Operations} We first analyze the cost of the flow sent from $v$ to the initial $W$, before we apply the remove operations. 

Focus on the cost incurred by type-(1) black components.  For a non-leaf or isolated root group  $G \in \calG$, the cost of sending the demands from $\union_{G' \in \children'(G), J \in \calJ^{(1)} \cap G'} J$ to $W \cap \union_{G' \in \children'(G)} U(\calJ^{(1)} \cap G')$ can be bounded as follows. We take an arbitrary representative $v$ at the root component of $G$, as the transfer point. First we send all the demands to $v$, and then we dispatch the demands from $v$ to $W \cap\union_{G' \in \children'(G)}  U(\calJ^{(1)} \cap G')$. The cost incurred in the two steps is at most
	\begin{align*}
		&\quad \sum_{G' \in \children'(G), J \in \calJ^{(1)} \cap G', i \in U(J)} \Big(y^*_iu'_i \cdot O(\ell^2) L_J  + \tilde y_iu'_i (d(i, J) + O(\ell^2) L_J )\Big).
	\end{align*}
    
Therefore, the cost of the initial flow incurred by type-(1) black components is
\begin{align*}
	M_1 := 
    \sum_{J\in \calJ^{(1)}, i\in U(J)}u_i'd(i, J)\tilde y_i+ O(\ell^2)\cdot\sum_{J \in \calJ^{(1)}\setminus \calJ_0, i \in U(J)} u_i' L_J(y_i^*+\tilde y_i).
\end{align*}

By $o_J(t)L_J\leq 2\Psi_J(t)$, for every $J\in \calJ^{(2, 3, 4)}$, then $\Psi_J(t)+O(\ell)o_J(t) L_J\leq O(\ell)\Psi_J(t)$.
And for $J\in \calJ^{(4)}$, since $\tilde y_J\in [t_J-1, t_J]$, we have $\Psi_J(\ceil{\tilde y_J})\leq \Psi_J(t_J)+(t_J-\tilde y_J)\Psi_J (t_J-1)$.
So the cost of initial flow incurred by type-(2), (3) and (4) black components is at most
\begin{align*}
M_{234} &:=\sum_{J \in \calJ^{(2, 3)}} \left(\Psi_J(t_J) + O(\ell)\cdot o_J(t_J) L_J\right) + \sum_{J \in \calJ^{(4)}} \left(\Psi_J(\ceil{\tilde y_J}) +O(\ell)\cdot o_J(\ceil{\tilde y_J})L_J \right)\\
&\leq O(\ell)\cdot \sum_{J \in \calJ^{(2, 3, 4)}} \Psi_J(t_J)  +O(\ell)\cdot \sum_{J \in \calJ^{(4)}}  (\Psi_J(t_J - 1)(t_J - \tilde y_J))
\end{align*}

The cost of the initial flow is at most $M_1 + M_{234}$.

\paragraph{Cost due to Removing Operations} 
Now we analyze the cost incurred by removing operations. During the analysis, we only bound the cost of the flows between representatives. To be more specific, suppose we resend the flow sent to some $i$ to another facility $i'$, where $i \in U_v$ and $i' \in U_{v'}$. Then we only consider the cost of sending the flow from $v$ to $v'$.  In our analysis, all the flow sent to $i$ are sent through $v$. That is, the cost of sending a unit flow from some representative $v''$ to $i$ is calculated as $d(v'', v) + d(v, i)$. This applies to the solutions for the penalty instances $S_J(t)$'s. Therefore, to resend the flow sent to $i$, we actually undo a portion of the flow sent from $v$ to $i$.  Second, due to the proportional distribution, we can charge the cost of sending the flow from $v'$ to $i'$ using the existing flow sent to $i'$.  We will call the cost of sending the flow from $v$ to $v'$ the $R$-to-$R$ cost, and the cost of sending the flow from $v'$ to $i'$ the $R$-to-$F$ cost. 

Notice that as we handle the black components $J$ in Algorithm~\ref{alg:reduce-W} in a top-to-bottom fashion, no facilities have been removed from $W \cap U(J)$ before we handle $J$. Then, we prove the following lemma:
\begin{lemma}
   The remove$(G, J)$ incurs a $R$-to-$R$ cost of at most $0$ if $J\in \calJ_0$,
   $O(\ell^2)\sum_{i\in U(J)} u'_iL_J \tilde y_i$ if $J\in \calJ^{(1)}\setminus \calJ_0$, and $O(\ell^3)\Psi_J(t_J - 1)$ if $J\in \calJ^{(3)}$. 
\end{lemma}
\begin{proof}
	Consider the case where $J$ is a type-(1) component. If $J = \{v\}$ is a singleton root component, then $G = \{J\}$. Also, $i^* \in U_v$. So, the $R$-to-$R$ cost of resending flow from $i^*$ to the new $W \cap U_v$ is $0$. 

    Now consider $J\in\mathcal J^{(1)}\setminus\mathcal J_0$.
Initially, the total demand assigned to facilities in $U(J)$ is at most $\sum_{i\in U(J)}u'_i\tilde y_i$. Let $i^*$ is a facility with the minimum demand in $W\cap U(G\cup\{J\})$, we have
$\text{demand}(i^*)
\le\min_{i\in W\cap U(J)}\text{demand}(i)
\leq \sum_{i\in U(J)}u'_i\tilde y_i.$
By Lemma \ref{le:distant-G-G'} and the triangle inequality, the
representative-to-representative distance used in the reassignment is at most $O(\ell^2)L_J$. Therefore, the $R$-to-$R$ cost is at most
$O(\ell^2)\sum_{i\in U(J)}u'_iL_J\tilde y_i.$

    The case where $J$ is type-(3) component,  we have $o_J(t_J-1)\geq \frac{1}{\ell} x^*(U(J), C)$, implying $\Psi_J(t_J-1) \geq \frac {L_J}{2\ell} \cdot x^*(U(J), C)$. Combining this inequality with the $O(\ell^2)L_J$ per-unit reassignment bound gives an $R$-to-$R$ cost of at most $O(\ell^3) \Psi_J(t_J-1)$.
\end{proof}
The remove operations incur a $R$-to-$R$ cost of 
\begin{align*}
	&\quad O(\ell^2) \sum_{J \in \calJ^{(1)}\setminus \calJ_0, i\in U(J)} u'_i L_J\tilde y_i + O(\ell^3)\sum_{J \in \calJ^{(3)}:\tilde y_J = t_J - 1} \Psi_J(t_J - 1)\\
	&\leq O(\ell^2) \sum_{J \in \calJ^{(1)}\setminus \calJ_0, i\in U(J)} u'_i L_J\tilde y_i  +  O(\ell^3)\sum_{J \in \calJ^{(3)}} (t_J - \tilde y_J) \Psi_J(t_J - 1).
\end{align*}

So, the final cost without counting the $R$-to-$F$ cost of the removing operations is at most
\begin{align}
&M_1 + M_{234} +O(\ell^2) \sum_{J \in \calJ^{(1)}\setminus \calJ_0, i\in U(J)} u'_i L_J\tilde y_i  +  O(\ell^3)\sum_{J \in \calJ^{(3)}} (t_J - \tilde y_J) \Psi_J(t_J - 1)\nonumber\\
&\leq \sum_{J\in \calJ^{(1)}, i\in U(J)}u'_i d(i, J)\tilde y_i+O(\ell^2)\sum_{J\in \calJ^{(1)}\setminus \calJ_0, i\in U(J)}u'_iL_J (y^*_i+\tilde y_i)
+O(\ell)\sum_{J\in \calJ^{(2, 3, 4)}}\Psi_J(t_J)\nonumber\\
&\qquad +O(\ell^3)\sum_{J\in \calJ^{(3, 4)}}(t_J-\tilde y_J)\Psi_J (t_J-1) \nonumber\\
&\leq  \sum_{J\in \calJ^{(1)}, i\in U(J)}u'_i d(i, J) y^*_i+O(\ell^2)\sum_{J\in \calJ^{(1)}\setminus \calJ_0, i\in U(J)}u'_iL_J y^*_i
+O(\ell)\sum_{J\in \calJ^{(2, 3, 4)}}\Psi_J(t_J)\nonumber\\
&\qquad +O(\ell^3)\sum_{J\in \calJ^{(3, 4)}}(t_J-y^*(U(J)))\Psi_J (t_J-1)\nonumber \\
&\leq  O(1)D(F)+O(\ell^3)\sum_{J\in \calJ^{(1)}\setminus \calJ_0}\Psi_J(t_J)+O(\ell^3)\sum_{J\in \calJ^{(3, 4)}}(t_J-y^*(U(J)))\Psi_J (t_J-1)+O(\ell)\sum_{J\in \calJ^{(2, 3, 4)}}\Psi_J(t_J)\nonumber\\
& \leq O(1)D(F)+O(\ell^3)\sum_{J\in \calJ\setminus \calJ_0}\Psi_J(t_J)+O(\ell^3)\sum_{J\in \calJ^{(3, 4)}}(t_J-y^*(U(J)))\Psi_J (t_J-1)\nonumber\\
&\leq O(1)\sum_{j \in C} \phi^*_j + O(\ell^6) \sum_{J \in \calJ, v \in J, j \in C}g_{j, v}(\phi^*_j + d(v, j)) \leq O(\ell^6)\sum_{j \in C}\phi^*_j.\nonumber
\end{align}
The second inequality holds by Lemma \ref{lem:feasible- LPauxiliary} and the third inequality holds by Lemma \ref{lemma:x*iCdiv}.

To see the fourth-to-last inequality, we bound $\Psi_J(t_J)$ for $J \in \calJ\setminus \calJ_0$ and $(t_J - y^*(U(J)))\Psi_J(t_J - 1)$ for $J \in \calJ^{(3, 4)}$ respectively. By inequality~\eqref{equ:all-constraint-J-t}, we have
\begin{align*}
\frac{t + 1-y^*(U(J))}{t +  1} \Psi_J(t) \leq (2|J|-1)\sum_{v \in J,j \in C}g_{j,v}(\phi^*_j + d(v, j))
\end{align*}
 Recall that $t_J = \floor{(1+\frac1{c\ell^2})y^*(U(J))}$. We have $\frac{y^*(U(J))}{t_J +  1} \leq 1/(1+\frac1{c\ell^2})$, implying that $\frac{t_J + 1-y^*(U(J))}{t_J +  1} \geq \Omega(1/\ell^2)$. Therefore, we have $\Psi_J(t_J) \leq O(\ell^3)\sum_{v \in J,j \in C}g_{j,v}(\phi^*_j + d(v, j))$.
 
 Consider the case $J \in \calJ^{(3, 4)}$, where we have $t_J > y^*(U(J))$.  
 \begin{align*}
	 &\quad (t_J - y^*(U(J))) \Psi_J(t_J - 1) = t_J \cdot \frac{t_J - y^*(U(J))}{t_J}\Psi_J(t_J - 1) \leq t_J  (2|J|-1)\sum_{v \in J,j \in C}g_{j,v}(\phi^*_j + d(v, j))\\
	 & \leq O(\ell^2)\sum_{v \in J,j \in C}g_{j,v}(\phi^*_j + d(v, j)),
 \end{align*}
where we used that $t_J \leq O(\ell)$.  \smallskip

Then we argue about the capacity violation of the solution, along with the $R$-to-$F$ cost of the removing operations.  In the initial flow, every facility $i \in W$ receives at most $(1+O(1/\ell))u_i$ units of flow.  As we argued, $|W \cap U(G)| \geq \Omega(\ell)$ always holds. When we send some flow to $W \cap U(G)$ for some non-leaf group or isolated root group $G \in \calG$, we guarantee that the amount of flow sent is at most the least flow received by any facility in $W \cap U(G)$. Also we apply the operation for any $G$ at most $O(1)$ times. Therefore, the operations will increase the flow received by any $i$ by a factor of at most $1 + O(1/\ell)$. This also implies that the $R$-to-$F$ cost incurred by the removing operations is negligible.  

Notice that the initial moving step in Section~\ref{Initial-Moving-of-Demands}, which moves demands from $J$ to $R$ has cost $O(1)\cdot \sum_{j \in C}\phi^*_j$. So, the overall moving cost of our solution is at most $O(\ell^6)\cdot \sum_{j \in C}\phi^*_j$. This finishes the proof of Theorem~\ref{thm:main-cap}.
     \bibliographystyle{plain}
    \bibliography{refs}
    \appendix
     \section{Omitted proofs}\label{sec:appendix}  

        \lemmarepresentatives*
         \begin{proof}
            For Property\ref{represen-distance}, let $v, v'\in R$ be distinct, and assume
            $d_{\texttt av}(v)\leq d_{\texttt av} (v')$.  When $v$ is added to $R$, the algorithm removes every client  $j$ with $d(v, j)\leq 4 d_{\texttt av}(j)$. Since $v'\in R$, it was not removed in that iteration. Hence $d(v, v')> 4 d_{\texttt av}(v')$.  For Property \ref{j-removed}, fix any client $j\in C$, and consider the iteration in which $j$ is removed from $C$.   Then consider Property \ref{opening-of-U_v}. By  \ref{represen-distance}, we have $B := \{ i \in F : d(i, v) \leq 2d_{\texttt av}(v) \} \subseteq U_v$. Since $d_{av}(v) = \sum_{i \in F} x^*_{i, v} d(i, v)$ and $\sum_{i \in F} x^*_{i, v} = 1$, we have $d_{\texttt av}(v) \geq (1 - x^*({B, v})) 2d_{\texttt av}(v)$, implying $y^*({U_v}) \geq y^*(B) \geq x^*({B, v}) \geq 1/2$, due to Constraint \eqref{LPC:x<y}. The case $d_{\mathrm{av}}(v)=0$ is trivial.

            Then we consider Property \ref{le:move-distance}. By \ref{j-removed}, there is a client $v' \in R$ such that $d_{\texttt av}(v') \leq d_{\texttt av}(j)$ and $d(v', j) \leq 4d_{\texttt av}(j)$. Since $d(i, v) \leq d(i, v')$ as $v' \in R$ and $i$ was added to $U_v$, we have $d(i, v) \leq d(i, v') \leq d(i, j) + d(j, v') \leq d(i, j) + 4d_{\texttt av}(j)$.
        \end{proof}

     \claimcomponentJ*
     \begin{proof}
    Fix a black component $J\subsetneq R$, and let $e'=(u, w)$ be the first edge processed by Kruskal’s algorithm with $u\in J$ and $w\notin J$. So $d(u, w)=L_J$. Since no previously added edge crosses the cut $(J, R\setminus J)$, the endpoints of $e'$ belong to distinct connected components, so $e'$ is added to $E_{\text{MST}}$.

Since $J$ is a final black component, $e'$ cannot be black. Thus, the merged component is heavy and remains heavy thereafter. Consequently, no subsequent black edge can enlarge the black component containing $u$, so this black component must already be $J$ before $e'$ is added. Therefore, every black edge $e=(v, v')$ in $J$ is added before $e$, implying $d(v, v')\le d(u, w)=L_J.$
For the case $J=R$,  the claim follows trivially.
\end{proof}      

    \lemmagrayforest*
\begin{proof}
    For Property  \ref{le:weight-J}, we recall that $y^*(F)=k\ge \ell$. Once a connected component maintained by Kruskal’s algorithm becomes heavy, it remains heavy. It cannot participate in any subsequent black-edge merge. Consequently, all black edges of a black component are added before any MST edge leaving it. Consider a light black component $J$. Since $y^*(F)\ge\ell$, we have $J\ne R$. Let $e$ be the first edge leaving $J$ that is added to the MST. Immediately before $e$ is added, $J$ is a connected component maintained by Kruskal’s algorithm. The edge $e$ is not white, since $J$ is light. It is not black, since $J$ is a maximal black component. Thus, $e$ is gray and directed away from $J$. After this merge, the connected component containing $J$ remains heavy. Hence, no subsequent gray edge is directed away from $J$. Therefore, every light black component has outdegree one in $\Upsilon$. Every heavy black component has outdegree zero, since it is already heavy before any edge leaving it is added. Now fix a gray tree $\tau\in\Upsilon$. Its underlying undirected graph is a tree, so
$\sum_{J\in V(\tau)}\deg_\tau^+(J) = |E(\tau)| = |V(\tau)|-1.$
Since every vertex has outdegree zero or one, exactly one vertex has outdegree zero. This vertex is a heavy black component. All other vertices are light black components, and all gray edges are directed towards this unique root. 

For Property \ref{length-gray}, let $e=(v, v')$ be the gray edge from $J\ni v$ to its parent $J'\ni v'$. Immediately before $e$ is added, $J $ is a connected component maintained by Kruskal’s algorithm. Hence $e$ is the first processed edge crossing $(J, R\setminus J)$, and therefore
$d(v, v')=d(J, R\setminus J)=L_J.$
Since $v\in J$ and $v'\in J'\subseteq R\setminus J$, we obtain $L_J\le d(J, J')\le d(v, v')=L_J,$ proving the desired equality.

For Property \ref{lem:leaf-to-root-original-gray}, every conponent $J$ and its parent $J'$ satisfies
$L_{J'}=d(J',R\setminus J')\le d(J',J)=L_J,$
where the last equality follows from \ref{length-gray}.
\end{proof}

 \lemmagroups*
 \begin{proof}
               Fix a group $G\in\mathcal G$. Let $r$ be the black
component selected when $G$ is constructed, and let
$r_\tau$ be the root of the gray tree $\tau$ containing $G$.
Set $P:=G\cup\{r\}$. The subgraph of $\tau$ induced by
$P$ is a tree. Let $\widetilde G$ be obtained by expanding
its nodes into their black components and retaining
the corresponding gray edges.
First consider the case $r\neq r_\tau$.
Since $r$ is not the root of $\tau$, by Lemma \ref{le:weight-J} its weight is less than $\ell$.
 By construction, the total weight of $P\setminus\{r\}$ is at most $2\ell$.
Thus the total weight of $P$ is at most $3\ell$.
By Lemma \ref{opening-of-U_v}, each original representative $v$ satisfies $y^*(U_v)\ge 1/2$.
Therefore the number of nodes in $\tilde{G}$ is at most $\frac{3\ell}{1/2}=6\ell$.

Now consider $r=r_\tau$. If $r$ has weight at most
$2\ell$, then $\widetilde G$ contains at most
$(2\ell+2\ell)/(1/2)=8\ell$ representatives.
Otherwise, $r$ is a singleton black component and
contributes only one representative. The remaining
components contribute at most $2\ell/(1/2)=4\ell$
representatives. Thus, $\widetilde G$ contains at most
$8\ell$ representatives in either case.

Fix $J\in G$, $v\in J$, and
$v'\in\text{root}_G$.
Their path in the original MST is contained in
$\widetilde G$.
If $r\in G$, then $\text{root}_G=r$.
Lemma~\ref{lem:leaf-to-root-original-gray} implies that every black component on
the path has $L$-value at most $L_J$.
Otherwise, let $K$ be the selected child of $r$ whose
subtree contains $J$. By the choice of
$\text{root}_G$, $L_{\text{root}_G}\le L_K\le L_J.$
Together with Lemma~\ref{lem:leaf-to-root-original-gray}, this gives the same bound
for every black component on the path.
By Claim~\ref{claim:componentJ} and Lemma~\ref{length-gray}, every edge on the path
has length at most $L_J$. Therefore, $d(v, v')\leq (8\ell-1)L_J=O(\ell)L_J.$
            \end{proof}
            
\lemmablackcomponent*

 \begin{proof}
            

        For Property \ref{lem:leaf-to-root-path}, for any child-parent edge $(J, J')$ in $\Upsilon^*$, we consider two cases: if $J'$ is the parent of $J$ in the original gray tree. By Lemma \ref{lem:leaf-to-root-original-gray}, the leaf-to-root monotonicity of $L$-values for the original gray tree, we have $L_{J'} \leq L_J$. If $J$ and $J'$ are adjacent siblings in the original tree, with $J'$ ordered before $J$. Due to the sorting rule, $L_{J'} \leq L_J$. Hence the $L$-values are non-increasing along every leaf-to-root path.


        For Property \ref{le:|J|=1}, suppose a black component $J$ contains multiple representatives, that is $|J| > 1$. Then it must be formed by merging at least two black components through the last black edge.
      Let the two child components before this final merge be $J_v$ and $J_{v'}$, satisfying $y^*(U(J_v)) < \ell$ and $y^*(U(J_{v'})) < \ell$.
      After merging, we have $J = J_v \cup J_{v'}$, and the total weight satisfies
      $y^*(U(J)) = y^*(U(J_v)) + y^*(U(J_{v'})) \leq 2\ell$.

        For Property \ref{lem:the-distance-of-J-and-its-parent}, 
        let $J''$ be the parent of $J$ in the original
   gray forest $\Upsilon$. By the construction of $\Upsilon^*$, either $J'=J''$, or $J'$ is the sibling of $J$ immediately preceding it in the non-decreasing order of $L$-values. Hence $L_{J''}\le L_J$ and $L_{J'}\le L_J$. For every $J\in \calJ$, $|J|\leq O(\ell)$, and by Claim \ref{claim:componentJ}, we have $\text{diam}(J)\leq O(\ell)L_J$.  If $J'=J''$, then the original gray edge between $J$ and $J'$ has length $L_J$. Thus, for every $v\in J$ and $v'\in J'$, $d(v,v') \leq \text{diam}(J)+L_J+\text{diam}(J') = O(\ell)L_J$. 
   Otherwise, both $J$ and $J'$ are children of $J''$ in $\Upsilon$. The gray edges from $J$ and $J'$ to $J''$ have lengths $L_J$ and $L_{J'}$, respectively. Hence, $d(v,v') \leq \text{diam}(J) +L_J +\text{diam}(J'') + L_{J'} +\text{diam}(J') = O(\ell)L_J$. Therefore, $\max_{v\in J, v'\in J'}d(v, v')\leq O(\ell)L_J$.
        \end{proof}

     \lemmaGproperties*
           \begin{proof}
            For Property \ref{lem-G:not a root group}, any non-root component $J$ has $y^*(U(J)) < \ell$. Thus, if $G$ is not a root group, the terminating condition constructing $G$ implies that $G$ had total weight less than $\ell$ right before the last black component was added to it. Then we have $\sum_{J\in G} y^*(U(J)) < 2\ell$.

              For Property \ref{lemG: a non-leaf group}, when we construct the group $G$ from the tree $\tau^*$ by Algorithm \ref{alg:construct-groups}, the terminating condition is that $G$ contain all descendants of $J$ in $\tau^*$ or $\sum_{J\in G} y^*(U(J)) \geq \ell$. Thus, if $G$ is not a leaf-group, then the condition that $G$ contains all descendants of $J$ in $\tau^*$ does not hold. So we have $\sum_{J\in G} y^*(U(J)) \geq \ell$.

           For Property \ref{lem-G:isolated-root},  
             let $J$ be the root black component of the transformed gray tree that contains $G$. By Lemma \ref{le:weight-J}, $J$ is heavy. Hence, $y^*(U(J))\ge \ell.$
And by  Algorithm \ref{alg:construct-groups},  the
root group remains $G=\{J\}$. 

        For Property \ref{lem-G:the-amount-children-of-G}, since the forest $\Upsilon^*$ is binary and every group $G\in \mathcal G$ contains at most $O(\ell)$ components, we have that every group $G$ contains at most $O(\ell)$ children.  

         Finally for Property \ref{le:distant-G-G'},  in the transformed gray tree $\Upsilon^*$, the $L$-values are non-increasing along every leaf-to-root path. By Algorithm \ref{alg:construct-groups}, we select the child component with the smallest $L$-value add $G$ each time.  Let $G'$ be a child group of G and let $J_0'$ be the root component of $G'$. Then $L_J \leq L_{J_0'}$ for every $J\in G$. On the other hand, every component $J'\in G'$ is a descendant of $J_0'$. By the monotonicity of L-values along root-to-leaf paths, $L_{J_0'}\le L_{J'}$ for all $J'\in G'$. Therefore, $\max_{J\in G} L_J \leq L_{J_0'} \leq \min_{J'\in G'} L_{J'}$.
         The property then holds by Claim \ref{claim:componentJ}, Lemma  \ref{lem:leaf-to-root-path}, \ref{lem:the-distance-of-J-and-its-parent} and that each group has at most $O(\ell)$ components. If $G'=G$, then $G=\{J'\}$ is a root group.
   Claim \ref{claim:componentJ} gives $\text{diam}(J')\le (|J'|-1)L_{J'}=O(\ell)L_{J'}$, which also implies the stated bound.
        \end{proof}

\end{document}